\documentclass[
notitlepage
,floatfix,
aps,
prx,
reprint,
onecolumn,
superscriptaddress
]{revtex4-1}
\usepackage[utf8]{inputenc}

\usepackage[caption=false]{subfig}
\usepackage{graphicx} 
\usepackage{epstopdf}
\usepackage{array}
\usepackage{verbatim}
\usepackage{amsmath,amsfonts,amssymb,amscd,mathtools}
\usepackage{amsthm}
\usepackage{tabularx}
\usepackage{stmaryrd}
\usepackage{enumerate}
\usepackage[ruled]{algorithm2e}
\usepackage{wasysym}
\usepackage{braket}
\usepackage{mathrsfs}
\usepackage{microtype}
\usepackage{hyperref}
\usepackage[dvipsnames]{xcolor}
\hypersetup{
    bookmarksnumbered=true, 
    unicode=false, 
    pdfstartview={FitH}, 
    pdftitle={}, 
    pdfauthor={}, 
    pdfsubject={}, 
    pdfcreator={}, 
    pdfproducer={}, 
    pdfkeywords={}, 
    pdfnewwindow=true, 
    colorlinks=true, 
    linkcolor=NavyBlue, 
    citecolor=NavyBlue, 
    filecolor=NavyBlue, 
    urlcolor=NavyBlue 
}

\usepackage[centering,hmargin=1.1in,vmargin=.75in]{geometry}

\usepackage{newtxtext,newtxmath}

\theoremstyle{plain}
\newtheorem{thm}{Theorem}
\newtheorem{cor}[thm]{Corollary}

\theoremstyle{definition}

\newcommand{\eq}[1]{(\hyperref[eq:#1]{\ref*{eq:#1}})}

\renewcommand{\sec}[1]{\hyperref[sec:#1]{Section~\ref*{sec:#1}}}
\newcommand{\thrm}[1]{\hyperref[thrm:#1]{Theorem~\ref*{thrm:#1}}}
\newcommand{\lemm}[1]{\hyperref[lemm:#1]{Lemma~\ref*{lemm:#1}}}
\newcommand{\prop}[1]{\hyperref[prop:#1]{Proposition~\ref*{prop:#1}}}
\newcommand{\corr}[1]{\hyperref[corr:#1]{Corollary~\ref*{corr:#1}}}
\newcommand{\fig}[1]{\hyperref[fig:#1]{~\ref*{fig:#1}}}
\newcommand{\deff}[1]{\hyperref[deff:#1]{~\ref*{deff:#1}}}

\newcommand{\mA}{\mathcal{A}}
\newcommand{\mC}{\mathcal{C}}
\newcommand{\mE}{\mathcal{E}}

\newcommand{\mT}{\mathcal{T}}

\newcommand{\mF}{\mathcal{F}}

\newcommand{\mM}{\mathcal{M}}

\newcommand{\mO}{\mathcal{O}}

\newcommand{\mK}{\mathcal{K}}

\newcommand{\mbF}{\mathbb{F}}
\newcommand{\mbI}{\mathbb{I}}
\newcommand{\mbM}{\mathbb{M}}

\newcommand{\mV}{\mathcal{V}}
\newcommand{\mW}{\mathcal{W}}
\renewcommand{\*}{\textup{*}}
\DeclareMathOperator{\cone}{cone}
\newcommand{\cleq}{\preceq}
\newcommand{\cgeq}{\succeq}
\newcommand{\cle}{\prec}
\newcommand{\cge}{\succ}

\DeclareMathAlphabet{\matheu}{U}{eus}{m}{n}

\DeclareMathOperator{\Tr}{Tr}

\newcommand{\ketbra}[2]{|{#1}\rangle\!\langle{#2}|}

\newcommand{\ba}{\begin{eqnarray}}
\newcommand{\ea}{\end{eqnarray}}
\newcommand{\bann}{\begin{eqnarray*}}
\newcommand{\eann}{\end{eqnarray*}}
\newcommand{\bal}{\begin{equation}\begin{aligned}}
\newcommand{\eal}{\end{aligned}\end{equation}}
\newcommand{\dm}[1]{\ketbra{#1}{#1}}

\newcolumntype{L}[1]{>{\raggedright}p{#1}}
\newcolumntype{C}[1]{>{\centering}p{#1}}
\newcolumntype{R}[1]{>{\raggedleft}p{#1}}
\newcolumntype{D}{>{\centering\arraybackslash}X}

\renewcommand{\*}{\textup{*}}
\newcommand{\<}{\left\langle}
\renewcommand{\>}{\right\rangle}
\DeclarePairedDelimiter\abs{\lvert}{\rvert}%
\DeclarePairedDelimiter\norm{\|}{\|}%
\let\V\mV

\newcommand{\sbar}{\;\rule{0pt}{9.5pt}\right|\;}

\newcommand{\lset}{\left\{\left.}
\newcommand{\rset}{\right\}}

\begin{document}
\title{General Resource Theories in Quantum Mechanics and Beyond:\\Operational Characterization via Discrimination Tasks}

\begin{abstract}
We establish an operational characterization of general convex resource theories --- describing the resource content of not only states, but also measurements and channels, both within quantum mechanics and in general probabilistic theories (GPTs) --- in the context of state and channel discrimination. We find that discrimination tasks provide a unified operational description for quantification and manipulation of resources by showing that the family of robustness measures can be understood as the maximum advantage provided by any physical resource in several different discrimination tasks, as well as establishing that such discrimination problems can fully characterize the allowed transformations within the given resource theory.

Specifically, we introduce a quantifier of resourcefulness of a measurement in any GPT, the generalized robustness of measurement, and show that it admits an operational interpretation as the maximum advantage that a given measurement provides over resourceless measurements in all state discrimination tasks. In the special case of quantum mechanics, we connect discrimination problems with single-shot information theory by showing that the generalized robustness of any measurement can be alternatively understood as the maximal increase in one-shot accessible information when compared to free measurements. We introduce two different approaches to quantifying the resource content of a physical channel based on the generalized robustness measures, and show that they quantify the maximum advantage that a resourceful channel can provide in several classes of state and channel discrimination tasks. Furthermore, we endow another measure of resourcefulness of states, the standard robustness, with an operational meaning in general GPTs as the exact quantifier of the maximum advantage that a state can provide in binary channel discrimination tasks. Finally, we establish that several classes of channel and state discrimination tasks form complete families of monotones fully characterizing the transformations of states and measurements, respectively, under general classes of free operations. Our results establish a fundamental connection between the operational tasks of discrimination and core concepts of resource theories --- the geometric quantification of resources and resource manipulation --- valid for all physical theories beyond quantum mechanics with no additional assumptions about the structure of the GPT required. 
\end{abstract}

\author{Ryuji Takagi}
\email{rtakagi@mit.edu}
\affiliation{Center for Theoretical Physics and Department of Physics, Massachusetts Institute of Technology, Cambridge, Massachusetts 02139, USA}
\author{Bartosz Regula}
\email{bartosz.regula@gmail.com}
\affiliation{School of Physical and Mathematical Sciences, Nanyang Technological University, 637371, Singapore}
\affiliation{Complexity Institute, Nanyang Technological University, 637335, Singapore}

\maketitle


\section{Introduction}

The advantages provided by quantum phenomena in the transfer and processing of information allowed for the technological boom currently transforming areas such as communication, computation, cryptography, and sensing \cite{dowling_2003,acin_2018}. The realization that intrinsic physical properties of quantum mechanics can be regarded precisely as \textit{resources} in information processing tasks sparked an investigation of quantum information in the so-called resource-theoretic setting, aiming to establish the theoretical and practical methods to characterize both the advantages and the limitations associated with different physical properties of quantum systems, measurements, and transformations \cite{chitambar1806quantum}. Such resource theories are now commonplace in the study of a diverse range of phenomena, such as entanglement~\cite{plenio2007introduction,HOrodecki_review2009}, coherence~\cite{aberg2006quantifying,Baumgratz2014,Streltsov2017}, asymmetry~\cite{Gour2008,Marvian2016}, quantum thermodynamics~\cite{Brandao2013,Brandao_secondlaws2015}, steering~\cite{gallego_2015}, non-Markovianity~\cite{Rivas2010non-Markov,bhattacharya2018resource,wakakuwa2017operational}, magic~\cite{Veitch2014,howard_2017}, non-Gaussianity~\cite{Genoni2008,Takagi2018,albarelli2018resource}, measurement simulability and incompatibility~\cite{oszmaniec_2017,oszmaniec2018all,guerini_2017}, measurement informativeness~\cite{skrzypczyk2018robustness}, and quantum memory of channels~\cite{rosset_2018}.

Although applications of the resource-theoretic framework have enhanced systematic studies of many physical settings and significantly contributed to a deeper understanding of our capabilities in manipulating and exploiting such resources, one could wonder whether the generality of the framework allows one to go beyond specific examples and obtain results applicable to a broad class of settings, thus providing a unified picture of resources in general. 
A complete study of which features are universal among all resources, stemming from only the very foundations of quantum mechanics, therefore remains a major area of investigation, and such an approach of \textit{general resource theories} has recently gained much attention \cite{horodecki_2012,brandao_2015,delrio_2015,coecke_2016,Liu2017,Gour2017,anshu_2017,regula_2018,lami_2018,takagi2018operational,chitambar1806quantum,li_2018}. Indeed, although the framework of resource theories began with the characterization of the properties of quantum states, it has recently been adapted to the study of quantum channels~\cite{Pirandola2017fundamental,bendana_2017,gour_2018-1,wilde_2018,rosset_2018,Diaz2018usingreusing,theurer_2018,li_2018,Zhuang2018non-Gaussian,seddon2019quantifying} and measurements~\cite{heinosaari_2015,haapasalo_2015,guerini_2017,carmeli_2018,skrzypczyk2018robustness}, allowing for the description of dynamic resources on a similar footing to static ones and thus motivating the question of whether all such resources can be described by a unified formalism.

In fact, one can pose an even more fundamental question: can common features of resource theories be understood without relying on quantum mechanics at all? Despite the success of quantum mechanics, the ongoing search for an axiomatic theory of probability and correlations in nature has provided us with insight into physical theories beyond quantum, as well as allowed for a straightforward unification of the methods required to characterize physical theories including classical and quantum probability theory. The formalism of general probabilistic theories (GPTs) \cite{ludwig_1985,hartkamper_1974,davies_1970,lami_2018-2} lends itself perfectly to the investigation of states, measurements, and their transformations at a fully fundamental and general level. It is particularly suited to illuminate which assumptions and which basic features of a theory lead to operational consequences, allowing one to identify the exact axioms one has to accept in order to recover the features of quantum theory~\cite{barrett_2007,chiribella_2010,chiribella_2011,barnum_2011,masanes_2011,barnum_2014,lee_2015,lee_2018}. This leads us to extend the framework of resource theories to general probabilistic theories and investigate a unified characterization of general resources in the extensive formalism of GPTs rather than limiting it to quantum mechanics, as has been previously considered for specific examples of resource theories~\cite{Chiribella_2015entanglement,Chiribella_2015diagonalize,Chiribella2016entanglement,Chiribella_2017microcanonical,Scandolo_thesis}.

As the very word ``resource'' suggests, understanding the operational aspects of resources --- how they can be utilized for physical tasks, and what limitations a resource theory places on the conversion of physical resources --- has central importance both theoretically and practically. 
However, it frequently requires resource-specific approaches and does not easily generalize to encompass all physically relevant resource theories, and it is thus highly desired to find a fundamental class of operational tasks that would allow for the understanding of the resourcefulness of a given physical property in general settings. 
A promising candidate for such a class of operational tasks which, on the one hand, lie at the heart of the non-classical features of quantum theory~\cite{HOLEVO1973337,Helstrom,kitaev_1997,chefles_2000-1,childs_2000,acin_2001-1,bae2015quantum,jencova_2014,watrous_2018} as well as GPTs~\cite{ludwig_1985,hartkamper_1974,kimura_2010,bae_2016-1,lami_2017} and, on the other hand, have found relevance in several existing resource theories, are the tasks of state and channel discrimination. 
In particular, Piani and Watrous~\cite{Piani2009} first showed that for every entangled quantum state, there exists a channel discrimination task in which it is more useful than any separable state. 
Similar results were subsequently found in several different resource theories of states~\cite{Piani2015,Piani2018,Napoli2016,Piani2016,takagi2018operational} and measurements~\cite{carmeli_2018,skrzypczyk2018robustness,skrzypczyk_2019}, and the work of Takagi et al.~\cite{takagi2018operational} finally showed that this property is shared by any convex resource theory of quantum states. It remains to understand how general this property truly is, and whether all resources --- both static and dynamic, both within quantum mechanics and beyond --- can provide explicit advantages in such operational tasks.

A fundamental aspect of any resource theory is its \textit{quantification}, which aims to measure the amount of inherent resources contained in a given physical object and allows for a quantitative comparison with other objects within the resource theory.
This can be approached in many inequivalent ways, and a plethora of possible choices of resource measures exist~\cite{regula_2018,chitambar1806quantum}. A natural question which arises in this context is whether the given measure can be understood in an operational sense, assessing exactly the usefulness of a given object in some physical task; however, establishing such an interpretation for a given quantifier is often highly nontrivial. The family of so-called robustness measures~\cite{vidal_1999,regula_2018} stands out in this context, as two prominent members of the family have found several applications in operational settings: these are the \textit{generalized robustness} \cite{Napoli2016,Piani2016,Bu2017,Regula2018oneshot,datta_2009,brandao_2010,brandao_2011,zhao_2018,berta_2017-1,anshu_2017,Piani2015,Piani2018,takagi2018operational} and the \textit{standard robustness} \cite{brandao_2011,brandao_2007,howard_2017}. 
They are not only fundamental resource quantifiers faithfully capturing the resourcefulness of given objects with  clear geometric interpretations, but also significantly relevant to experiments --- they are directly observable, that is, can be obtained in an experiment by measuring a single, suitably chosen observable, rather than requiring complicated and expensive methods such as state tomography, allowing for the experimental quantification of resources \cite{Brandao2005,eisert_2007}.
In particular, the investigation of discrimination tasks in resource theories of states~\cite{Piani2015,Piani2018,Napoli2016,Piani2016,takagi2018operational} revealed a notable similarity: the advantage that a given resource provides in such discrimination tasks is often quantified precisely by the generalized robustness. The generality of this quantitative relation was unveiled in Ref.\,\cite{takagi2018operational}, which showed that it is true in every convex resource theory of quantum states. This, together with recent progress in the resource theories of measurements which showed a similar interpretation of robustness quantifiers for measurements in specific settings~\cite{skrzypczyk2018robustness,skrzypczyk_2019}, suggests that a unified operational interpretation of generalized robustness which could account for dynamic resources in addition to static ones might be possible. However, no such universal interpretation of any of the robustness measures has been obtained thus far --- and in fact, despite several known applications of the generalized robustness in the context of discrimination, it has not been known whether the standard robustness plays a role in understanding such problems whatsoever.

The other predominant problem which resource theories are expected to tackle is the \textit{manipulation} of resources, which asks whether it is possible to transform one resource to another when constrained to only employ the transformations allowed within the given resource theory. It is particularly insightful to understand this question in relation with the operational and quantitative aspects of resource theories --- does there exist a family of resource measures or operational tasks which completely characterizes resource manipulation? The work of Skrzypczyk and Linden~\cite{skrzypczyk2018robustness} indicated a potential of discrimination-type problems in this respect by showing that a family of state discrimination tasks fully characterize the simulability of measurements by classical post-processing. A comprehensive extension of this type of an operational characterization to more general settings which, together with quantification, would complete an operational characterization of general resource theories, has hitherto remained elusive.

\subsection{Summary of results}

In this work, we solve the problems raised above under the umbrella of operational tasks of discrimination --- specifically, we characterize general convex resource theories of states, measurements, and channels, establishing tools for their quantification, endowing the class of robustness measures with an explicit operational interpretation as the advantage that a physical object can provide in various discrimination tasks, and showing that such discrimination tasks fully characterize the conversion between states or measurements with free operations of the given resource theory. The generality of our methods and results establishes a universal operational description of resources in general probabilistic theories (GPTs), revealing strong connections between several aspects of general resources and showing that the underlying convex structures can provide deep insight into the properties of physical systems also in an operational setting. We stress that all of our results are immediately applicable to broad classes of physically relevant quantum resource theories of states (including entanglement, coherence, magic, asymmetry, athermality\ldots), measurements (informativeness, simulability, separable and positive partial transpose measurements\ldots), and channels (quantum memories, free channels in the resource theories of states\ldots).

We begin by providing an introduction to the main concepts of general probabilistic theories, discrimination problems, and resource quantification in Sec.\,\ref{sec:prelim}.

Our investigation starts in Sec.\,\ref{sec:gen_rob_states} by extending the results of Ref.\,\cite{takagi2018operational} beyond quantum mechanics, and showing that the generalized robustness of states in any convex resource theory and any GPT can be understood as the quantifier of the advantage that a given resourceful state can provide in channel discrimination problems.

We then introduce the measurement robustness in Sec.\,\ref{sec:gen_rob_measurements}, which quantifies the resource content of any measurement in convex resource theories of measurements, generalizing the recent approach of Ref.\,\cite{skrzypczyk2018robustness} to arbitrary resources and probabilistic theories. We in particular establish a direct operational interpretation of the measurement robustness in any resource theory by showing that it quantifies exactly the maximum advantage that a measurement can provide in all state discrimination tasks compared to all resourceless measurements. 

Having established the connection between discrimination tasks and generalized robustness of measurements, in Sec.\,\ref{sec:accessible} we further extend this connection to single-shot information theory within the setting of quantum mechanics.
We specifically show that the increase in min-accessible information, a single-shot variant of accessible information, of an ensemble of states effected by applying resourceful measurements as compared to free measurements is exactly quantified by the generalized robustness.

The considerations extend beyond the case of states and measurements. In Sec.\,\ref{sec:gen_rob_channels}, we formalize the quantification of the resource content of channels in two different ways: by measuring the amount of a resource that a channel can generate by acting on a free state, as well as in a more abstract formalism of convex resource theories of channels. 
In particular, we introduce resource generating power, as well as generalized robustness of channels defined in general resource theories.
We establish operational interpretations of these different robustness measures by showing that the resource generating power of a channel exactly characterizes the advantage that the channel can enable in resourceless-state discrimination tasks, and by showing that, within quantum theory, the robustness of a given channel or the maximum robustness of a given ensemble of channels quantifies the advantage it provides in a class of state and channel discrimination tasks, respectively.

We additionally consider in Sec.\,\ref{sec:standard_rob} another resource measure in the robustness family, the standard robustness, and show that it admits a universal operational interpretation in any convex resource theory of states in the context of quantifying the advantage that a state provides over resourceless states in all balanced binary channel discrimination tasks.  This gives a general operational meaning to this quantity for the first time, extending the link between resource quantification and operational advantages in discrimination tasks to the standard robustness of states defined in general settings.

Finally, in Sec.\,\ref{sec:complete_monotones} we show that different classes of discrimination tasks can form complete sets of monotones in any resource theory, in the sense that a state or a measurement can be transformed into another state or measurement using only free operations of the given resource theory if and only if the former performs better than the latter in a family of channel or state discrimination tasks. We therefore reveal explicitly that discrimination tasks can be useful not only in the context of quantifying the resource strength, but also in fully characterizing the conversions between resource states or measurements. We also extend the results to describe the transformations of state ensembles.

Besides operationally characterizing two main concepts of resource theories, quantification and manipulation of resources, our results not only endow the robustness measures with direct operational interpretations in general resource theories of states, measurements, and channels, but also explicitly demonstrate strong relations between several seemingly unrelated notions --- discrimination-type tasks, geometric resource quantification, resource transformations, and one-shot quantum information-theoretic quantities (see Fig.\,\ref{fig:concept}) --- as conjectured in Ref.\,\cite{skrzypczyk2018robustness} for quantum mechanics. The majority of our results apply to every single physical probabilistic theory in finite dimensions, relying only on the convex structure of the underlying resources and requiring no assumptions about the structure of the GPT beyond the basic axiom referred to as the no-restriction hypothesis~\cite{ludwig_1985,barrett_2007,chiribella_2010}. We make use of methods in convex analysis and in particular conic programming.

\begin{figure}[htbp]
    \centering
    \includegraphics[width=0.75\columnwidth]{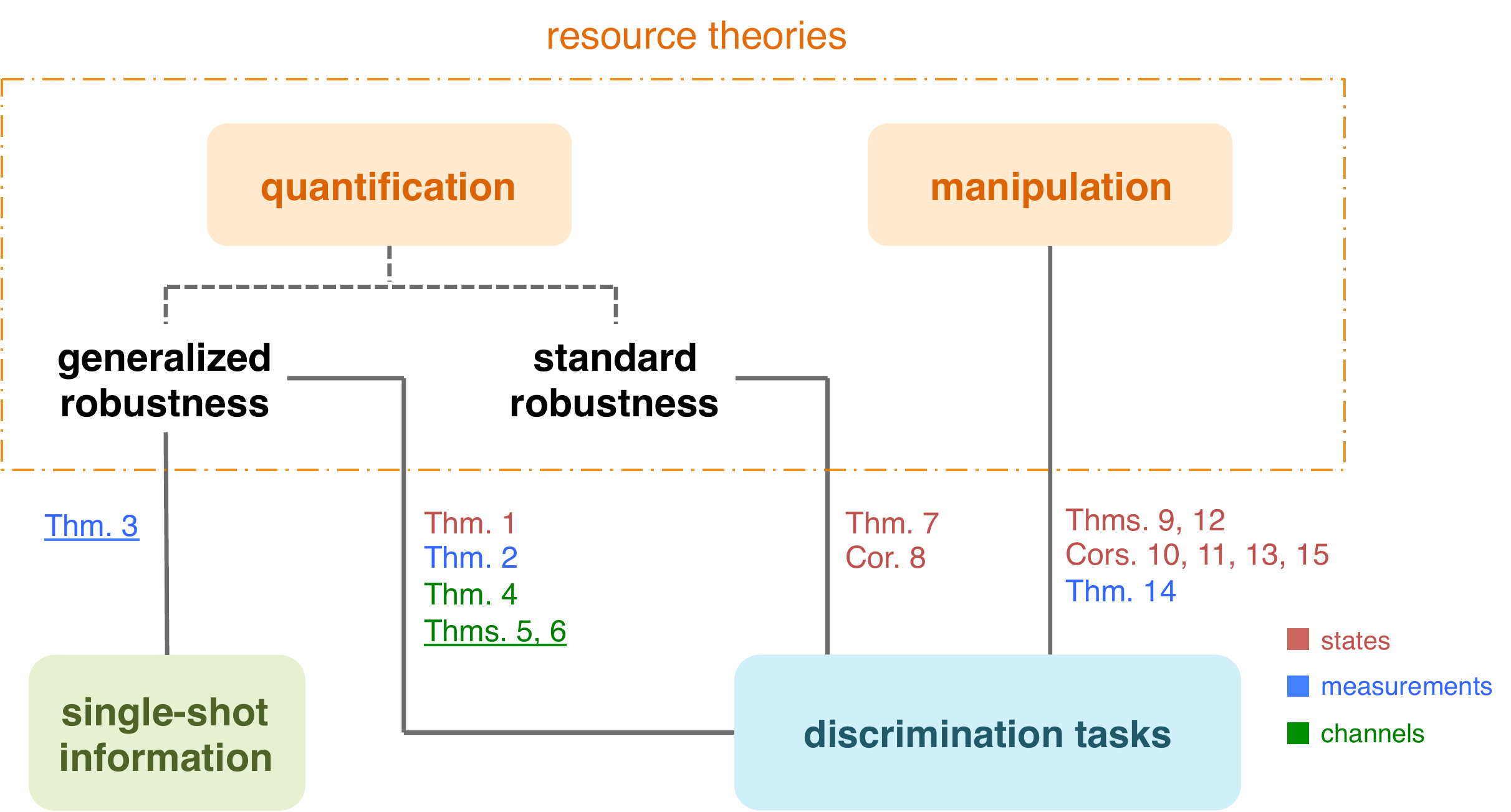}
    \caption{Schematics showing the connections established by the results in this work. The colors of the Theorem labels indicate the types of resources relevant to the results (red: states, blue: measurements, green: channels). The underlined Theorems (Thms.\,\ref{thm:accessible measurement}, \ref{thm:gen_rob_single_channel}, and \ref{thm:gen_rob_channel_ensemble}) are shown specifically for quantum theory, while the other results are valid in general GPTs.}
    \label{fig:concept}
\end{figure}


\section{Preliminaries}\label{sec:prelim}


\subsection{General probabilistic theories}


We will briefly outline the basic formalism of general probabilistic theories. We refer the interested reader to Ref.~\cite[Ch.\,1-2]{lami_2018-2} which provides a modern introduction to the topic and a detailed discussion, derivation, and justification of the concepts.

The physical setting of a GPT can be identified with a convex and closed set of \textit{states} $\Omega(\mV)$ in a finite-dimensional real complete normed vector space $\mV$ and a set of \textit{effects} contained in the dual vector space $\mV\*$, which correspond to the results of physically implementable measurements. 

A crucial role will be played by the cone generated by $\Omega(\mV)$, $\mC \coloneqq \lset \lambda\, \omega \sbar \lambda \in \mathbb{R}_+,\; \omega \in \Omega(\mV) \rset$, which we will further assume to be pointed (i.e. $\mC \cup (-\mC) = \{0\}$) and generating (i.e. $\operatorname{span} \mC = \mV$), such that it induces a partial order on the space $\mV$ given by $x \cleq_\mC y \iff y - x \in \mC$, with $x \cle_\mC y \iff y - x \in \operatorname{int}(\mC)$. The dual cone $\mC\* \coloneqq \lset E \in \mV\* \sbar \< E, \omega \> \geq 0 \; \forall \omega \in \mC \rset$, where we write $\< E, \omega \>$ for $E(\omega)$, then similarly induces a partial order $E \cleq_{\mC\*} F \iff F - E \in \mC\*$ on the dual space.

Associated with each GPT is a fixed \textit{unit effect} $U \cge_{\mC\*} 0$, defined as the unique element of the dual space satisfying $\< U, \omega \> = 1$ for all $\omega \in \Omega(\mV)$; equivalently, this allows one to understand the set of valid states as the set of normalized elements of the cone $\mC$ in the sense that $\Omega(\mV) = \lset \omega \in \mC \sbar \<U, \omega \> = 1 \rset$. Under the so-called no-restriction hypothesis, which we will hereafter take as an assumption, the effects are then all functionals $E$ such that $0 \cleq_{\mC\*} E \cleq_{\mC\*} U$, which are precisely all linear functionals $E: \Omega(\mV) \to [0,1]$. Any finite collection of effects $\{M_i\}_i$ such that $\sum_i M_i = U$ will be called a \textit{measurement}, with $\< M_i, \cdot \>$ identifying the probability of measuring the $i$\,th outcome. We will denote by $\mM$ the set of all possible measurements.

Given two GPTs defined by the spaces $\mV$ and $\mV'$ with the corresponding sets of states $\Omega(\mV)$ and $\Omega(\mV')$, one can then study the transformations betweeen them.
The question of physically allowed transformations $\Lambda: \mV \to \mV'$ between different states, which we will refer to as \textit{channels}, is in general heavily dependent on the physical setting of the given theory and additional assumptions placed on it~\cite{edwards_1971,edwards_1972}. However, since we require the output to be a valid state, two general assumptions can be made: (1) any valid channel $\Lambda$ is state cone-preserving, that is, $\Lambda[\mC] \subseteq \mC'$ where $\mC'$ is the cone defined by $\Omega(\mV')$ in the output space; and (2) $\Lambda$ is normalization-preserving, that is $\< U, x \> = \< U', \Lambda(x) \>$ for any $x \in \mC$ where $U'$ is the unit effect in the output dual space. Any valid set of physical transformations in the given GPT, which we will denote by $\mT(\mV, \mV')$, will necessarily be a subset of all state cone- and normalization-preserving operations, although often the inclusion will be strict. To allow for full generality of our results, we will place no further restrictions on the set of physically allowed transformations at this point, except for the trivial assumption that the identity map $\mathrm{id}: x \mapsto x$ is physically implementable.

For any $\Lambda: \mV \to \mV'$, define the dual map $\Lambda\*: \mV'\* \to \mV\*$ by
\begin{equation}\begin{aligned}\label{eq:map_duality}
  \< E, \Lambda(x) \> = \< \Lambda\*(E), x \> \;\, \forall E \in \mV'\*,\, x \in \mV.
\end{aligned}\end{equation}
We will, without loss of generality, identify the bidual $\V\*\*$ with $\V$ and say that $\Lambda\*\* = \Lambda$.

Consider now transformations between measurements, that is, maps $\Gamma: \mV'\* \to \mV\*$ such that for any measurement $\{M_i\}_i$, $\{\Gamma(M_i)\}_i$ is also a measurement. It is not difficult to see that two conditions need to be satisfied for such a map to always result in valid measurements: one, it needs to preserve the effect cone in the sense that $\Gamma[\mC'\*] \subseteq \mC\*$, and two, it needs to preserve the unit effect, $\Gamma(U')=U$, so that $\sum_i \Gamma(M_i) = U$. We will refer to maps obeying the two conditions as effect cone-preserving and unital, respectively.  Using the duality relation \eqref{eq:map_duality}, it is then not difficult to see that the set of all maps dual to the set of state cone- and normalization-preserving maps are precisely the effect cone-preserving unital maps; more specifically, a map $\Lambda$ is state-cone preserving iff its dual preserves the effect cone, and normalization-preserving iff its dual is unital.

The concepts of quantum theory can be intuitively understood in this setting. To help translate the notation of the quantum mechanical setting to more general GPTs, we have included a basic comparison in Table \ref{tab:gpts}.


\subsection{Discrimination tasks}


Consider a finite ensemble of the form $\{p_i, \sigma_i\}_{i}$, where $\sigma_i \in \Omega(\mV)$ and $p_i$ are the probabilities corresponding to each state, such that $\sum_i p_i = 1$. The task of \textit{state discrimination} is concerned with the scenario where a state is sampled from the ensemble and one aims to determine which one of the states is in one's possession by measuring it. Specifically, given a measurement $\{M_i\}_{i}$, we associate the $i$\,th measurement outcome with the guess that the sampled state is $\sigma_i$. The average probability of successfully obtaining the correct guess is then given by
\begin{equation}\begin{aligned}
  p_{\rm succ} \left( \{p_i, \sigma_i\}, \{M_i\} \right) = \sum_i p_i \< M_i, \sigma_i \>.
\end{aligned}\end{equation}

In the context of minimum-error discrimination, one is in particular interested in choosing an optimal measurement which maximizes this quantity. It is a fundamental fact in any GPT that, generalizing the approach in the seminal Holevo--Helstrom theorem \cite{HOLEVO1973337,Helstrom}, in the case of discriminating between two states the problem can be expressed as~\cite{kimura_2010}
\begin{equation}\begin{aligned}
  \max_{\{M_i\} \in \mM} p_{\rm succ} \left( \{p_i, \sigma_i\}_{i=0}^{1}, \{M_i\}_{i=0}^1 \right) = \frac12 \left( \norm{p_0 \sigma_0 - p_1 \sigma_1}_{\Omega} + 1 \right)
\end{aligned}\end{equation}
where $\norm{\cdot}_{\Omega}$ is the so-called \textit{base norm}, given by
\begin{equation}\begin{aligned}
  \norm{x}_{\Omega} &\coloneqq \min \lset \lambda_+ + \lambda_- \sbar x = \lambda_+ \omega_+ - \lambda_- \omega_-,\; \lambda_{\pm} \in \mathbb{R}_+,\; \omega_{\pm} \in \Omega(\mV) \rset\\
  &= \max \lset \< E, x \> \sbar -U \cleq_{\mC\*} E \cleq_{\mC\*} U \rset,
\end{aligned}\end{equation}
where the equality follows by convex duality \cite{rockafellar2015convex}. Notice that there exists a measurement which distinguishes an ensemble of two states perfectly if and only if $\norm{p_0 \sigma_0 - p_1 \sigma_1}_\Omega = 1$. In the case that $p_0 = p_1 = \frac12$, we will refer to this task as balanced binary discrimination.

\begin{table}
\caption{A comparison between the concepts and standard nomenclature used in finite-dimensional quantum mechanics and in more general GPTs.}
\label{tab:gpts}
\begin{tabular}{l @{\qquad} l}
\toprule
\textrm{GPTs} & \textrm{Quantum mechanics} \\
\colrule
Real vector space $\mV$ & Self-adjoint operators acting on a Hilbert space $\mathcal{H}$ \\
States $\Omega(\mV)$ & Density operators \\
State cone $\mC$ & Positive semidefinite operators \\
Effect cone $\mC\*$ & Positive semidefinite operators \\
Unit effect $U$ & Identity operator $\mbI$ \\
Canonical bilinear form $\< E, x \>$ & Hilbert-Schmidt inner product $\Tr(E x)$\\
Measurement & Positive operator-valued measure (POVM) \\
Effect & POVM element\\
State cone-preserving maps & Positive maps \\
Effect cone-preserving maps & Positive maps \\
Normalization-preserving maps & Trace-preserving maps \\
Physical transformations $\mT(\mV, \mV')$ & Completely positive trace-preserving maps\\
Unital maps & Unital maps\\
Base norm $\norm{\cdot}_\Omega$ & Trace norm $\norm{\cdot}_1$\\
Order unit norm $\norm{\cdot}^\circ_\Omega $ & Operator norm $\norm{\cdot}_\infty$\\
\botrule
\end{tabular}
\end{table}

We will furthermore make frequent use of the norm dual to the base norm $\norm{\cdot}_\Omega$, called the order unit norm $\norm{\cdot}^\circ_\Omega$, which can be obtained as
\begin{equation}\begin{aligned}
  \norm{Y}^\circ_{\Omega} &= \max \lset \< Y, x \> \sbar \norm{x}_{\Omega} \leq 1 \rset\\
  &= \max \lset \abs{\< Y, \omega \>} \sbar \omega \in \Omega(\mV) \rset. 
\end{aligned}\end{equation}
Notice in particular the set of effects is precisely $\lset Y \in \mC\* \sbar \norm{Y}^\circ_\Omega \leq 1 \rset$.

Another setting often encountered in the task of state discrimination is where one is constrained to use only a restricted set of allowed measurements $\{M_i\} \in \mM_\mF \subseteq \mM$. This scenario has been found to be of fundamental importance due to the phenomenon of data hiding, that is, the existence of states which can be distinguished perfectly with general measurements but not with local measurements supplemented with classical communication \cite{terhal_2001,divincenzo_2002}. Provided that the set $\mM_\mF$ is informationally complete, that is, the effects contained in $\mM_\mF$ together span the whole space $\mV\*$, the best success probability in this setting can be expressed as \cite{matthews_2009-1,lami_2017}
\begin{equation}\begin{aligned}
  \sup_{\{M_i\} \in \mM_\mF} p_{\rm succ} \left( \{p_i, \sigma_i\}_{i=0}^{1}, \{M_i\}_{i=0}^1 \right) = \frac12 \left( \norm{p_0 \sigma_0 - p_1 \sigma_1}_{\mM_\mF} + 1 \right)
\end{aligned}\end{equation}
with $\norm{\cdot}_{\mM_\mF}$ being the so-called distinguishability norm $\norm{x}_{\mM_\mF} \coloneqq \sup_{\{M_i\} \in \mM_\mF} \sum_i \abs{\< M_i, x \>}$.

A related task is concerned with \textit{channel discrimination}, where one of the channels from a given ensemble $\{p_i, \Lambda_i\}$ with each $\Lambda_i \in \mT(\mV, \mV')$ occurring with probability $p_i$ is applied to a known state $\omega \in \Omega(\mV)$. The task then is, by measuring the output state $\Lambda_i(\omega)$, to decide which of the channels was applied. The average probability of guessing correctly with a measurement $\{M_i\}$ is then
\begin{equation}\begin{aligned}
  p_{\rm succ} \left( \{p_i, \Lambda_i\}, \{M_i\}, \omega \right) = \sum_i p_i \< M_i, \Lambda_i(\omega) \>
\end{aligned}\end{equation}
which is completely equivalent to discriminating the state ensemble $\{p_i, \Lambda_i(\omega)\}$. A more general setting is that of \textit{subchannel discrimination}, in which the object to discriminate is an ensemble $\{\Psi_i\}$ of subchannels, that is, state cone-preserving maps which are normalization-nonincreasing in the sense that $\< U, \omega \> \geq \< U', \Psi_i(\omega) \>$, and their sum $\sum_i \Psi_i$ is normalization-preserving. We can then write
\begin{equation}\begin{aligned}
  p_{\rm succ} \left( \{\Psi_i\}, \{M_i\}, \omega \right) = \sum_i \< M_i, \Psi_i(\omega) \>.
\end{aligned}\end{equation}

Finally, in some cases we will allow for the \textit{inconclusive} discrimination, that is, a discrimination task in which an ensemble $\{p_i, \sigma_i\}_{i=0}^{N-1}$ is discriminated with a measurement $\{M_i\}_{i=0}^{N}$, and we associate with the $N$th measurement outcome an inconclusive result. Similarly, the average probability of guessing correctly will then be
\begin{equation}\begin{aligned}
   p'_{\rm succ} \left( \{p_i, \sigma_i\}_{i=0}^{N-1}, \{M_i\}_{i=0}^{N} \right) = \sum_{i=0}^{N-1} p_i \< M_i, \sigma_i \>.
\end{aligned}\end{equation}
and analogously for the case of channel and subchannel discrimination.
Inconclusive discrimination is a broader class of tasks than conclusive discrimination, in the sense that the latter can be considered as a special case of the former where one takes $M_N=0$.  


\subsection{Resources and their quantification}

A general resource theory can be identified with a set of objects (here: a set of states, measurements, or channels) together with a set of transformations of said objects that one deems \textit{free}, in the sense that they are available in the given physical setting at no resource cost. Both definitions will generally depend on the physical setting under consideration --- intuitively, an object being free can be understood as it possessing no resource, and a transformation being free means that it is allowed within the given physical constraints. In particular, any such free operation should not generate any resource; that is, any object subjected to a free transformation should have quantitatively ``less'' of the given resource than it possessed initially. This is formalized precisely by \textit{resource monotones}, which are functions from the set of states, measurements, or channels to real numbers whose aim is to quantify the resource content of the given object, and therefore do not increase under the action of the free transformations. Frequently, further constraints are imposed on functions admissible as valid resource monotones~\cite{vidal_2000,chitambar1806quantum,regula_2018}, although we will make no additional assumptions at this point. In most cases, the choice of a resource monotone is not unique, and typically one therefore looks for a choice of monotones which characterize the physical properties of the given resource in addition to merely outputting a number associated with an object.

Having defined the setting of a resource theory in this way, one is then interested in understanding the limitations that it places on one's operational capabilities within the general probabilistic theory. Some of the fundamental questions that can be asked in this context are: what physical advantages a resourceful object can provide over a free one; which transformations are possible with the restricted set of free operations, and how to characterize them; and what exactly can measuring the resourcefulness of an object tell us quantitatively about the usefulness of the object in physical tasks.

The concept which will form a central pillar of this work is convexity. Although ultimately a technical assumption, it stems from deep physical considerations, and is a foundation of any GPT~\cite{ludwig_1983}. Take in particular the set of states $\Omega(\mV)$: if one is free to prepare any $\omega_1, \omega_2 \in \Omega(\mV)$ within the given physical setting of the theory, one should also be allowed to simply forget which one was prepared --- such a randomization leads precisely to convex mixtures $p \omega_1 + (1-p) \omega_2$. Dually, this means that allowed measurements should form a convex set, and analogously the randomization of transformations should be a valid physical procedure. In a similar way, this property can be explicitly required from a given resource theory; if a randomization procedure of free objects is free to be performed, the associated set of objects should be free. Although it is certainly possible to define resource theories in which convexity does not hold~\cite{Modi2012,Genoni2008,liu2017diagonal}, it is a natural assumptions in most physical settings, and the vast majority of established resource theories are indeed convex.

In any such convex resource theory, a very intuitive way to define a quantifier is as follows: given an object $A$, what is the least amount of mixing $p \in [0,1]$ with another object $B$ such that $(1-p) A + p B$ is free? This is precisely the idea behind robustness measures~\cite{vidal_1999}, where in particular the \textit{standard (free) robustness} asks about the least coefficient $p$ such that $B$ is also a free object, and the \textit{generalized (global) robustness} asks about the least $p$ when $B$ is any admissible object. In an intuitive sense, this can be understood as the robustness of the resource contained in $A$ to noise in the form of admixing the objects $B$.

To investigate the properties of such monotones, we will now specify to the particular settings of convex resource theories of states, measurements, and channels in GPTs.



\section{Generalized robustness of states}\label{sec:gen_rob_states}


The generalized robustness has found a multitude of uses in quantum resource theories. Its operational applications include the tasks of  one-shot entanglement dilution \cite{brandao_2010,brandao_2011}, one-shot coherence distillation \cite{Bu2017,Regula2018oneshot} and dilution \cite{zhao_2018}, phase discrimination with coherent states \cite{Napoli2016,Piani2016}, and catalytic resource erasure \cite{berta_2017-1,anshu_2017}.

Notably, it was shown in Ref.\,\cite{takagi2018operational} that the generalized robustness of states defined in any quantum resource theory serves as an exact quantifier for the maximum advantage that a resource state provides in a class of (sub)channel discrimination tasks. 
However, the construction considered in the proof relies on the specific structure of quantum theory, and it does not immediately generalize to all GPTs.
Here, we introduce a generalization of that result which holds without any assumptions about the underlying GPT, showing that this universal relation applies even beyond the setting of quantum mechanics.

Given a convex and closed set of free states $\mF \subseteq \Omega(\mV)$, the generalized robustness is given for any state $\omega \in \Omega(\mV)$ as the optimal value of the convex optimization problem
\begin{equation}\begin{aligned}
  R_\mF(\omega) \coloneqq& \min \lset r \in \mathbb{R}_+ \sbar \frac{\omega+r \delta}{1+r} \in \mF,\; \delta \in \Omega(\mV) \rset\\
  =& \min \lset r \in \mathbb{R}_+ \sbar \omega \cleq_{\mC} (1+r)\ \sigma,\; \sigma \in \mF \rset.
\end{aligned}\end{equation}
To ensure that this quantity is well-defined for any state, we will assume that $\mF$ contains at least one interior point of $\mC$.
It is straightforward (see e.g. Appendix \ref{app:duality}) to obtain the following equivalent dual problem:
\begin{equation}\begin{aligned}
{\text{\rm maximize }}& \ \  \< X, \omega \> - 1  \\
{\text{\rm subject to }}&\ \ X\cgeq_{\mC^{\*}} 0\\
&\ \ \< X, \sigma \> \leq 1,\ \forall \sigma\in \mF.
\end{aligned}
\label{eq:gen rob state dual}
\end{equation}
One can check that the dual problem is strictly feasible by taking the feasible solution $X=U/2$, and thus strong duality is ensured by Slater's theorem~\cite{boyd_2004}, meaning that the solutions to the primal and dual problems coincide.

Recalling that the average probability of success in a channel discrimination task is given by
\begin{equation}\begin{aligned}
   p_{\rm succ} \left( \{p_i, \Lambda_i\}, \{M_i\}, \omega \right) = \sum_i p_i \< M_i, \Lambda_i(\omega) \>,
\end{aligned}\end{equation}
we can now show that the generalized robustness of a state quantifies its maximal advantage over the free states in all such channel discrimination tasks with a fixed choice of measurement.

\begin{thm}\label{pro:gen rob state}
For any $\omega \in \Omega(\mV)$ it holds that
\begin{equation}\begin{aligned}
  \max_{\{M_i\}, \{p_i,\Lambda_i\}} \frac{p_{\mathrm{succ}} (\{p_i,\Lambda_i\}, \{M_i\}, \omega)}{\max_{\sigma \in \mF}  p_{\mathrm{succ}} (\{p_i,\Lambda_i\}, \{M_i\}, \sigma)} = 1 + R_\mF(\omega) 
\end{aligned}\end{equation}
where the maximization is over all finite ensembles of channels $\{p_i,\Lambda_i\}$ with each $\Lambda_i \in \mT(\mV,\mV')$ and all measurements $\{M_i\}$ on the output space $\mV'$.
\end{thm}
\begin{proof}
To show that the left-hand side is upper bounded by the right-hand side, note that the definition of the robustness implies that, for any state $\omega$, there exists another state $\delta$ such that $\frac{\omega + r \delta}{1+r} = \sigma \in \mF$ where $r = R_\mF(\omega)$. This gives for any $\{M_i\}$ and any channel ensemble $\{p_i,\Lambda_i\}$ that
\begin{equation}\begin{aligned}
   p_{\mathrm{succ}} (\{p_i,\Lambda_i\}, \{M_i\}, \omega) &= \sum_i p_i\< M_i, \Lambda_i\left[ (1+r) \sigma - r \delta \right] \>\\
  &\leq \sum_i p_i\< M_i, (1+r) \Lambda_i(\sigma) \>\\
  &\leq (1+r) \max_{\sigma \in \mF} p_{\mathrm{succ}} (\{p_i,\Lambda_i\}, \{M_i\}, \sigma),
 \label{eq:gen rob state ineq first}
\end{aligned}\end{equation}
where the first inequality follows since each $M_i \in \mC\*$ and $\delta \in \mC$, which concludes the first part of the proof.

To see the opposite inequality, consider an optimal solution $X$ in \eqref{eq:gen rob state dual} for $R_\mF(\omega)$ and define a measurement by $\left\{\frac{X}{\norm{X}^\circ_\Omega},\,U - \frac{X}{\norm{X}^\circ_\Omega} \right\}$. This is a valid measurement since each $X \in \mC\*$ by the dual form of the robustness in  \eqref{eq:gen rob state dual}, and so $0 \cleq_{\mC\*} X/\norm{X}^\circ_\Omega \cleq_{\mC\*} U$ by definition of the norm $\norm{\cdot}^\circ_\Omega$. Consider now the channel ensemble defined by $p_0=1$, $\Lambda_0=\mathrm{id}$ and $p_1=0$, $\Lambda_1=\Lambda'$ where $\mathrm{id}$ denotes the identity map $x \mapsto x$, and $\Lambda'$ is an arbitrary channel. This gives
\begin{equation}\begin{aligned}
  \max_{\{M_i\}, \{p_i,\Lambda_i\}} \frac{ p_{\mathrm{succ}} (\{p_i,\Lambda_i\}, \{M_i\}, \omega)}{\max_{\sigma \in \mF}  p_{\mathrm{succ}} (\{p_i,\Lambda_i\}, \{M_i\}, \sigma)} &\geq \frac{\frac{1}{\norm{X}^\circ_\Omega} \< X, \omega \>}{\frac{1}{\norm{X}^\circ_\Omega} \max_{\sigma \in \mF} \< X, \sigma \>}\\
  &\geq 1 + R_\mF(\omega) 
  \label{eq:gen rob state ineq second}
\end{aligned}\end{equation}
where in the last inequality we used that $\< X, \sigma \> \leq 1$ for any free state $\sigma$ by the condition in Eq. \eqref{eq:gen rob state dual}.
\end{proof}
We stress that, although the above result optimizes over the set of all ensembles of \textit{physical} transformations $\mT(\mV,\mV')$, the only assumption about the set $\mT(\mV,\mV')$ we make is that it contains the identity map. This makes the above Theorem immediately applicable to every GPT, regardless of how the given physical limitations constrain the implementable set of operations.

We remark that instead of channel discrimination, one can alternatively consider subchannel discrimination --- which is generally a broader class of discrimination tasks --- and show the same relation, whose proof proceeds analogously. (Eq. \eqref{eq:gen rob state ineq first} still holds when one replaces the channel ensemble with a set of subchannels, and one can choose the same channel ensemble to show \eqref{eq:gen rob state ineq second}.)

Importantly, from the dual form in \eqref{eq:gen rob state dual} one can see that the robustness is directly observable, in the sense that it can be obtained by measuring a single effect $X$ (expectation value of observable $X$ for the case of quantum mechanics) at the state $\omega$. This ensures that the quantification of the robustness is accessible, allowing in particular to straightforwardly bound the value of $R_\mF$ based on measurement data obtained in experiment, adapting the approach of so-called quantitative resource witnesses~\cite{Brandao2005,eisert_2007}. One should also note that the generalized robustness can be computed exactly for certain classes of states in quantum resource theories such as entanglement and Schmidt number $k$ entanglement ~\cite{harrow_2003,Steiner2003,regula_2018}, coherence and multi-level coherence~\cite{Piani2016,ringbauer_2018,regula_2018,johnston_2018}, magic~\cite{regula_2018,bravyi_2018-1}, and in several cases can be cast as a semidefinite program for any state; this makes the computation of the operationally motivated quantity in Thm.~\ref{pro:gen rob state} feasible in practice for many relevant cases of resource theories.


\section{Generalized robustness of measurements}\label{sec:gen_rob_measurements}

Understanding the discriminative power of restricted sets of measurements is of central importance not only in characterizing the operational consequences precipitated by limitations of physically allowed measurements, but often also in studying the very fundamental structure of the underlying GPT~\cite{matthews_2009-1,lami_2017,aubrun_2018}. The phenomenon of data hiding \cite{terhal_2001,divincenzo_2002} has in particular motivated the study of the question: given a measurement, how well can one distinguish physical states with it as compared to some fixed restricted set of measurements? We will show that a robustness measure associated with the measurement can provide a precise answer to the question.




First, we formally define the generalized robustness of measurements with respect to some convex and closed set of measurements $\mM_\mF$, which we will define as
\begin{equation}\begin{aligned}
   \mM_\mF \coloneqq \lset \{M_i\}_i \in \mM \sbar M_i \in \mE_\mF \; \forall i \rset
 \end{aligned}\end{equation} 
where $\mE_\mF \subseteq \mC\*$ is some chosen convex and closed cone of free effects which we are able to access within the constraints of the given resource theory. As examples of such a setting, one can consider local measurements, separable measurements~\cite{lami_2017}, measurements simulable by a given set of measurements \cite{oszmaniec_2017,guerini_2017,filippov_2018}, or trivial measurements (proportional to the unit effect); within quantum mechanics, one can furthermore choose, for instance, positive partial transpose (PPT) measurements, incoherent measurements, or (probabilistic mixtures of) Pauli measurements. 

We define the generalized robustness of measurement with respect to $\mE_\mF$ for a given measurement $\mbM = \{M_i\}_i$ as
\ba 
 R_{\mE_{\mF}}(\mbM) \coloneqq \min \lset r \in \mathbb{R}_+ \sbar M_i +  r \, N_i\in \mE_{\mF} \; \forall i,\;\, \{N_i\}_i \in \mM \rset.
\label{eq:rob meas def}
\ea
We will assume that $\mM_\mF$ contains at least one measurement consisting of effects $M_i \cge_{\mC\*} 0$ which are in the interior of $\mC\*$, so that the above quantity is well-defined for any measurement. The faithfulness: $R_{\mE_{\mF}}(\mbM)=0$ iff $M_i \in \mE_{\mF}\, \forall i$, the convexity: $R_{\mE_{\mF}}(p\mbM+(1-p)\mbM')\leq pR_{\mE_{\mF}}(\mbM)+(1-p)R_{\mE_{\mF}}(\mbM')$, and the monotonicity of the robustness: $R_{\mE_{\mF}}(\Gamma(\mbM)) \leq R_{\mE_{\mF}}(\mbM)$ for any effect-cone preserving map $\Gamma$ s.t. $\Gamma[\mE_\mF] \subseteq \mE_\mF$ follow easily from the definition.
It is also straightforward to show the monotonicity under classical post-processing: $R_{\mE_{\mF}}(\mbM')\leq R_{\mE_{\mF}}(\mbM)$ holds where $M'_b=\sum_a p(b|a)M_a$ and $p(b|a)$ is any conditional probability distribution.
To see this, note that the definition together with the conic structure of $\mE_\mF$ implies that for any $i$, there exists a free effect $F_i\in \mE_{\mF}$ such that $\frac{M_i + r N_i}{1+r}=F_i$.
Thus, one can write for any $b$ and $p(b|a)$,
\ba
 M_b'&=&\sum_a p(b|a)\left((1+r)F_a-rN_a\right)\\
 &=& (1+r)F'_b - rN'_b
 \label{eq:rob meas decomposition}
\ea
where $F_a\in\mE_\mF$, $N_a\in C^{\*}$ are some effects. Note that $F'_b\coloneqq \sum_a p(b|a)F_a$, $N'_b\coloneqq \sum_a p(b|a)N_a$ are also effects constituting valid measurements, and $F'_b\in \mE_\mF$ due to the convexity of $\mE_\mF$. 
Since \eqref{eq:rob meas decomposition} is one valid decomposition of effects constituting $\mbM$, we get $R_{\mE_\mF}(\mbM')\leq R_{\mE_\mF}(\mbM)$ due to the minimization involved in the definition \eqref{eq:rob meas def}. 

To observe a close connection between the robustness and discrimination tasks, let us define $N'_i \coloneqq r N_i$, and rewrite the definition \eqref{eq:rob meas def} as the following convex optimization problem:
\ba
&{\text{\rm minimize}}& \ \ \lambda  \label{eq:opt_obj}\\
&{\text{\rm subject to}}&\ \ M_i + N'_i \cgeq_{\mE_{\mF}} 0 \label{eq:opt_cond1}\\
&& \ \ N'_i \in \mC\*\; \forall i\\
&& \ \ \sum_i N'_i = \lambda\, U.
\ea
An equivalent dual problem (see e.g. Appendix \ref{app:duality}) is written as
\ba
&{\text{\rm maximize}}& \ \ -\sum_i \< M_i, \sigma_i \>\label{eq:dual_robmes}  \\
&{\text{\rm subject to}}&\ \ \sigma_i \cleq_{\mC} \eta\;\; \forall i \label{eq:dual_cond2}\\
&& \ \ \eta\in\mV\\
&& \ \ \sigma_i \in \mE_{\mF}^*\;\; \forall i \label{eq:dual_cond1}\\
&& \ \ \<U, \eta\> = 1, \label{eq:dual_cond4}
\ea
and defining $\omega_i = - \sigma_i + \eta$, we can rewrite it as
\ba
&{\text{\rm maximize}}& \ \ \sum_i \< M_i, \omega_i \> - 1 \label{eq:dual2_opt}\\
&{\text{\rm subject to}}&\ \ \omega_i \in \mC\;\; \forall i\\
&& \ \ \eta\in\mV\\
&& \ \ \< F, \omega_i \> \leq \< F, \eta \> \;\; \forall i, \; \forall F \in \mE_\mF \label{eq:dual2_cond2}\\
&& \ \ \<U, \eta\> = 1 \label{eq:dual2_cond4}
\ea
where in \eqref{eq:dual2_opt} we used that $\sum_i \< M_i, \eta \> = \< U, \eta \> = 1$ by \eqref{eq:dual2_cond4}, and in \eqref{eq:dual2_cond2} we have written the condition $\eta - \omega_i \in \mE_\mF^*$ explicitly. To see that strong duality holds and thus the optimal value of the dual is equal to the optimal value of the primal problem, choose any $\sigma \cge_{\mC} 0$ (and therefore $\sigma \cge_{\mE_{\mF}^*} 0$) s.t. $0 < \< U, \sigma \> < 1$, which is guaranteed to exist since the interior of $\mC$ is nonempty by assumption, and define $\eta = \sigma / \< U, \sigma \>$. This choice of $\sigma_i = \sigma$ and $\eta$ can be noticed to strictly satisfy the conditions \eqref{eq:dual_cond2}-\eqref{eq:dual_cond1} and so Slater's theorem ensures that strong duality holds.

This form of the dual problem allows one to identify the generalized robustness of measurement as an exact quantifier for the advantage in some state discrimination task.
Let $\mA=\{p_i,\sigma_i\}$ denote a state ensemble to be discriminated. 
Then, we obtain the following result. 

\begin{thm}\label{thm:gen_rob_meas_as_advantage}
Let $\mM_\mF$ be the set of measurements whose effects are elements of $\mE_{\mF}$. Then, for any measurement $\mbM = \{M_i\}_i$ it holds that
 \ba
  \max_\mA \frac{p_{\rm succ}(\mA,\mbM)}{\max_{\mbF\in \mM_\mF}p_{\rm succ}(\mA,\mbF)} = 1 + R_{\mE_{\mF}}(\mbM)
 \ea
 where the maximization is over all finite ensembles of states $\mA=\{p_i,\sigma_i\}$.
\end{thm}
\begin{proof}
We first show that the left-hand side is smaller than or equal to the right hand side.
Following the definition of the generalized robustness, one can write $M_i = (1+r)F_i - rN_i$ for some $F_i \in \mE_{\mF}$, $N_i \in \mC\*$ for all $i$ where $r=R_{\mE_{\mF}}(\mbM)$.
Thus, we get
 \bal
  p_{\rm succ}(\mA, \mbM)&=\sum_i p_i \< M_i, \sigma_i \>\\
  &= (1+r)\sum_i p_i \<F_i, \sigma_i \> - r\sum_i p_i \< N_i, \sigma_i \>\\
  &\leq (1+r)\max_{\mbF\in \mM_\mF}p_{\rm succ}(\mA,\mbF),\label{eq:rob meas first}
 \eal
 where the inequality follows since $\sigma_i\in\mC$ and $N_i \in \mC\*$ for all $i$, which completes the first part of the proof.
 To show the converse, consider the set of optimal $\{\omega_i\} \subseteq \mC$ that appear in \eqref{eq:dual2_opt}. 
 We choose the probability distribution and states as
 \bal
 \begin{cases}
    p_i = \frac{\<U, \omega_i\>}{\sum_a \< U, \omega_a \>},\,\sigma_i= \omega_i/ \< U, \omega_i \> & {\rm when}\  \< U, \omega_i \> >0,\\
    p_i = 0 & {\rm when}\  \< U, \omega_i \> = 0.
 \end{cases}
 \label{eq:ensemble optimal}
 \eal
 
 Considering an ensemble defined by the above probability and states $\mA = \{p_i,\sigma_i\}$, we obtain for any $\mbF\in \mM_\mF$,
 \ba
  \frac{p_{\rm succ}(\mA,\mbM)}{p_{\rm succ}(\mA,\mbF)} &=& \frac{\frac{1}{\sum_a\ \< U, \omega_a \>}\sum_i \<M_i, \omega_i \>}{\frac{1}{\sum_a \< U, \omega_a \>}\sum_i \< F_i, \omega_i \>} \\
  &\geq& \frac{1+R_{\mE_\mF}(\mbM)}{\sum_i \< F_i, \eta \>} \\
  &=& \frac{1+R_{\mE_\mF}(\mbM)}{\< U, \eta \>} = 1+R_{\mE_\mF}(\mbM)
 \ea
 where the inequality is due to \eqref{eq:dual2_cond2} and the second last equality is due to \eqref{eq:dual2_cond4}, which concludes the proof.
\end{proof}

The above theorem establishes an explicit connection in any GPT between the inherent resourcefulness of a given measurement and the advantage realized in state discrimination tasks with respect to a general set of free measurements, which ensures an operational interpretation of the generalized robustness of measurements, extending the previously considered case of states. 
Furthermore, it allows for a connection with the data hiding phenomenon.
In the investigation of data hiding in general probabilistic theories, it is common to encounter the so-called data hiding ratio~\cite{matthews_2009-1,lami_2017}, which in our notation can be expressed as
\begin{equation}\begin{aligned}
  \max_{\tilde{\mA}} \frac{\max_{\mbM \in \mM}\; p_{\rm succ}(\tilde{\mA},\mbM) - \frac12}{\max_{\mbF\in \mM_\mF} p_{\rm succ}(\tilde{\mA},\mbF) - \frac12}
\end{aligned}\end{equation}
where the ensembles $\tilde{\mA}$ are limited to contain at most two different states. Thm.~\ref{thm:gen_rob_meas_as_advantage} then shows that maximizing the robustness $R_{\mE_\mF}(\mbM)$ over all measurements $\mbM$ provides an alternative ratio of this form, generalized to include state ensembles of arbitrary length.


\subsection{Connections with single-shot information theory}
\label{sec:accessible}

Here, we extend the connection between generalized robustness of measurements and state discrimination tasks to another seemingly different information-theoretic quantity, one-shot accessible information, which generalizes the specific case considered in Ref.\,\cite{skrzypczyk2018robustness}. 
Since entropic quantities are particularly relevant to quantum information theory, in this section we restrict our attention to quantum mechanics (and not general GPTs).

Consider the situation where Alice encodes her classical information into an ensemble $\mA=\{p_x,\sigma_x\}$, and Bob tries to decode it by making a POVM measurement on a state sampled from $\mA$. 
When this process is independently repeated for asymptotically many times, the amount of information he may learn is quantified by the accessible information  $I^{\rm acc}(\mA)\coloneqq \max_{\{N_y\}}I(X:Y)$ where $I(X:Y)=H(X)-H(X|Y)$ is the classical mutual information, $X$ is the random variable associated with the ensemble, and $\{N_y\}$ refers to a POVM measurement whose associated random variable is $Y$ \cite{wilde2013book}.
However, quantities based on the von Neumann/Shannon entropy cease to be suitable for more practical nonasymptotic cases, and several kinds of alternative quantities playing major roles in \textit{single-shot} scenarios have been proposed and studied \cite{Renner2004renyi,renner2008security,Ciganovic2014mutual}.
As in Ref.\,\cite{skrzypczyk2018robustness}, we consider a variant of single-shot version of accessible information, min-accessible information, for a state ensemble $\mA$ defined by  
\bal
I_{\min}^{\rm acc}(\mA)&:= \max_{\{N_y\} \in \mM} I_{\min}(X:Y)\\
    &= -\log \max_x p_x + \max_{\{N_y\} \in \mM}\log \sum_y \max_x p_x\Tr[\sigma_x N_y]
\label{eq:accessible def}
\eal
where $I_{\min}(X:Y)\coloneqq H_{\min}(X)-H_{\min}(X|Y)$ is a single-shot variant of mutual information \cite{Ciganovic2014mutual}, and  $H_{\min}(X)= -\log \max_x p_x$, $H_{\min}(X|Y)= -\log \sum_y\max_x p(x,y)$ are the min-entropy and min-conditional entropy \cite{Renner2004renyi,renner2008security,Konig2009min}.
We shall see that the accessible information for state-ensembles enables the information-theoretic characterization of the generalized robustness of measurements.
To see the relation, it is insightful to see the measurements as channels; in particular, consider the measure-and-prepare channel associated with the measurement $\mbM$ with POVMs $\{M_i\}$ defined by $\Lambda_\mbM(\cdot):= \sum_j \Tr[\cdot\, M_j]\dm{j}$.
Denoting the ensemble one would possess by applying the channel $\Lambda$ to the ensemble $\mA=\{p_x,\sigma_x\}$ by $\mA_{\Lambda_\mbM}:=\{p_x,\Lambda(\sigma_x)\}$, we obtain the following relation between the maximal increase in min-accessible information due to the given resource measurement and the generalized robustness of that measurement. 

\begin{thm} \label{thm:accessible measurement}
For any measurement $\mbM\in\mM$, it holds that
\ba
\max_{\mA}\left[I_{\min}^{\rm acc}(\mA_{\Lambda_\mbM}) - \max_{\mbM'\in \mathcal{M}_\mF} I_{\min}^{\rm acc}(\mA_{\Lambda_{\mbM'}})\right] = \log\left(1+R_{\mE_\mF}(\mbM)\right)
\label{eq:accessible rob meas}
\ea
where $\mathcal{M}_\mF$ is the set of free measurements consisting of POVM elements in $\mE_\mF$. 
\end{thm}
\begin{proof}
Consider the ensemble $\mA=\{p_x,\sigma_x\}$. Using the definition \eqref{eq:accessible def}, we get
\ba
 I_{\min}^{\rm acc}(\mA_{\Lambda_\mbM}) &=& -\log \max_x p_x + \max_{\{N_y\}}\log \sum_y \max_x \left[p_x\sum_a\Tr[\sigma_x M_a]\braket{a|N_y|a}\right]\\
 &=& -\log \max_x p_x + \log \sum_y \max_x \left[p_x\sum_a\Tr[\sigma_x M_a]\delta_{ay}\right]\\
  &=& -\log \max_x p_x + \log\max_{\{q(x|y)\}}\sum_y \sum_{x} q(x|y)\,p_x\Tr[\sigma_x M_y].
\ea
The second equality is because $N_y=\dm{y}$ can be chosen as optimal POVM elements, and the third equality is because $\max_x f(x)=\max_{\{q(x)\}}\sum_x q(x)f(x)$ for any function $f(x)$ and probability distribution $q(x)$. 
Note that we interchanged the summation over $y$ and the maximization over probability distributions $\{q(x|y)\}$ because the summation over $y$ is maximized when $q(x|y)$ maximizes the summant for each $y$. 
Then, we get for any ensemble $\mA$,
\ba
 I_{\min}^{\rm acc}(\mA_{\Lambda_\mbM}) - \max_{\mbM'\in \mathcal{M}_\mF} I_{\min}^{\rm acc}(\mA_{\Lambda_{\mbM'}}) &=& \log\frac{\max_{\{q(x|y)\}}\sum_y\sum_{x} q(x|y)\,p_x\Tr[\sigma_x M_y]}{\max_{\mbM'\in \mathcal{M}_\mF}\max_{\{q(x|y)\}}\sum_y\sum_{x} q(x|y)\,p_x\Tr[\sigma_x M_y']}\label{eq:accessible meas 1}\\
 &\leq& \log\left(1+R_{\mE_\mF}(\mbM)\right)
 \label{eq:accessible meas 1 ineq}
\ea
where the inequality can be proved in a similar way to \eqref{eq:rob meas first}. 
To see the converse, note that for any $\mbM'\in \mM_\mF$ and conditional probability distribution $q(x|y)$, another measurement $\mbM''$ defined by the POVM $M_x''=\sum_y q(x|y)M_y'$ is also a member of $\mM_\mF$ due to the convexity of $\mE_\mF$ that defines $\mM_\mF$.
Thus, \eqref{eq:accessible meas 1} is equivalently written as 
\ba
 I_{\min}^{\rm acc}(\mA_{\Lambda_\mbM}) - \max_{\mbM'\in \mathcal{M}_\mF} I_{\min}^{\rm acc}(\mA_{\Lambda_{\mbM'}}) &=& \log\frac{\max_{\{q(x|y)\}}\sum_y\sum_{x} q(x|y)\,p_x\Tr[\sigma_x M_y]}{\max_{\mbM'\in \mathcal{M}_\mF}\sum_{x} p_x\Tr[\sigma_x M_x']}\label{eq:accessible meas 2}.
\ea
The equality in \eqref{eq:accessible meas 1 ineq} can be achieved by setting $q(x|y)=\delta_{xy}$ and taking the ensemble defined by \eqref{eq:ensemble optimal}, which concludes the proof.

\end{proof}

Thm.\,\ref{thm:accessible measurement} generalizes the result in Ref.\,\cite{skrzypczyk2018robustness} for the case of theory of informativeness of measurements to general quantum resource theories of measurements. 
As a result, Thm.\,\ref{thm:accessible measurement}, together with Thm.\,\ref{thm:gen_rob_meas_as_advantage}, establishes the general connections between discrimination tasks, robustness measures, and single-shot information theory conjectured in Ref.\,\cite{skrzypczyk2018robustness} for the case of quantum measurements.


\section{Robustness measures for channels}\label{sec:gen_rob_channels}


Besides states and measurements, the state transformations themselves can also be regarded as resourceful, and their operational characterization in this way has recently become an active area of investigation \cite{bendana_2017,rosset_2018,Diaz2018usingreusing,gour_2018-1,theurer_2018,li_2018,seddon2019quantifying,wilde_2018}.

Here, we in particular discuss two general ways to approach the resource content of channels. One is to relate them to an underlying resource theory for states: if a channel $\Lambda \in \mT(\mV,\mV')$ is a free operation in a resource theory characterized by the set of free states $\mF$, then we know that $\Lambda[\mF] \subseteq \mF'$; the resourcefulness of a non-free operation can then be measured exactly by understanding how much resource it can create, which is known as the \textit{resource generating power}. The second approach is to define an arbitrary set of transformation $\mO_\mF$ which we deem as free, and quantify the resource content of a channel by defining a robustness measures in a similar way as we have done for states and measurements.


\subsection{Robustness generating power of channels}


Let $\mF \subseteq \Omega(\mV)$ be a set of free states in a given resource theory, with $\mF'$ the corresponding free states in another space $\mV'$. If a transformation $\Lambda \in \mT(\mV,\mV')$ is a free operation in this resource theory, we know for certain that $\sigma \in \mF \Rightarrow \Lambda(\sigma) \in \mF'$. As we have seen before, the resource content of a state is naturally and operationally quantified by the robustness $R_\mF(\omega)$ --- motivated by this, we would like to understand the best robustness achievable after the application of a channel on a resourceless state. The robustness generating power of a map $\Lambda$ is then defined as \cite{zanardi_2000,Diaz2018usingreusing}
\begin{equation}\begin{aligned}\label{eq:gen_pow}
  P_\mF(\Lambda) &\coloneqq \max_{\sigma \in \mF} R_{\mF'}(\Lambda(\sigma))\\
  &= \max \lset \< W, \Lambda(\sigma) \> -1 \sbar \sigma \in \mF,\; W \in \mC'\*,\; \< W, \pi \> \leq 1 \; \forall \pi \in \mF' \rset
\end{aligned}\end{equation}
where the second line follows from the dual form of the robustness.

We will now show that this quantity admits an operational interpretation in any GPT as the best advantage that $\Lambda$ can enable in state discrimination tasks of free-state ensembles  with a given measurement, where the transformation $\Lambda$ is applied to the ensembles prior to measurement. Specifically, consider a state discrimination task with the possibility of an inconclusive outcome, where the average success probability is
\begin{equation}\begin{aligned}
   p'_{\mathrm{succ}} (\{p_i, \sigma_i\}_{i=0}^{N-1}, \{M_i\}_{i=0}^N, \Lambda) \coloneqq \sum_{i=0}^{N-1} p_i \< M_i, \Lambda(\sigma_i) \>.
\end{aligned}\end{equation}
We then have the following.

\begin{thm}
Given a map $\Lambda \in \mT(\mV,\mV')$, its robustness generating power is equivalently given by
\begin{equation}\begin{aligned}
  \max_{\{M_i\}} \frac{\max_{\{p_i, \sigma_i\},\, \sigma_i \in \mF}  p'_{\mathrm{succ}} (\{p_i, \sigma_i\}, \{M_i\}, \Lambda)}{\max_{\{q_i, \pi_i\},\,\pi_i \in \mF'}  p'_{\mathrm{succ}} (\{q_i, \pi_i\}, \{M_i\}, \mathrm{id})} = 1 + P_\mF(\Lambda) 
\end{aligned}\end{equation}
where $\mathrm{id}$ denotes the identity map on $\mV'$.
\end{thm}
\begin{proof}
One direction is as usual. Notice that the definition of the robustness generating power implies that, for any free state $\sigma_i\in\mF$, there exists another state $\delta_i \in \Omega(\mV')$ such that $\frac{\Lambda(\sigma_i) + r_i \delta_i}{1+r_i} = \rho_i \in \mF'$ where $r_i \leq r = P_\mF(\Lambda)$. This gives for any $\{M_i\}$ and any free-state ensemble $\{p_i, \sigma_i\}$ that
\begin{equation}\begin{aligned}
   p'_{\mathrm{succ}} (\{p_i, \sigma_i\}, \{M_i\}, \Lambda) &= \sum_i p_i \< M_i, (1+r_i) \rho_i - r_i \delta_i \>\\
  &\leq \sum_i p_i \< M_i, (1+r) \rho_i \>\\
  &\leq (1+r) \max_{\substack{\{q_i, \pi_i\}\\\pi_i \in \mF'}} p'_{\mathrm{succ}} (\{q_i, \pi_i\}, \{M_i\}, \mathrm{id}).
\end{aligned}\end{equation}

On the other hand, consider the optimal solution $W$ in \eqref{eq:gen_pow} and define a measurement by $\left\{\frac{W}{\norm{W}_{\Omega'}^\circ},\,U' - \frac{W}{\norm{W}_{\Omega'}^\circ} \right\}$. We can then choose the single-element ensemble $\{1, \sigma\}$ where $\sigma$ is an optimal state in \eqref{eq:gen_pow}. This gives
\begin{equation}\begin{aligned}
  \max_{\{M_i\}} \frac{\max_{\{p_i, \sigma_i\},\, \sigma_i \in \mF}  p'_{\mathrm{succ}} (\{p_i, \sigma_i\}, \{M_i\}, \Lambda)}{\max_{\{q_i, \pi_i\},\,\pi_i \in \mF'}  p'_{\mathrm{succ}} (\{q_i, \pi_i\}, \{M_i\}, \mathrm{id})} &\geq \frac{\frac{1}{\norm{W}_{\Omega'}^\circ} \< W, \Lambda(\sigma) \>}{\frac{1}{\norm{W}_{\Omega'}^\circ} \max_{\pi \in \mF'} \< W, \pi \>}\\
  &\geq 1 + P_\mF(\Lambda) 
\end{aligned}\end{equation}
where in the last inequality we used that $\pi \in \mF'$ and so $\< W, \pi \> \leq 1$ by the conditions on $W$ in \eqref{eq:gen_pow}.
\end{proof}

This establishes a universal operational meaning of the robustness generating power for any resource theory and any choice of free operations in the given theory.

We remark that Ref.~\cite{li_2018} showed an operational interpretation in the context of binary channel discrimination of another measure of channel resourcefulness within quantum mechanics, the resource generating power as quantified by the trace norm distance to the set of free states: $\max_{\sigma \in \mF} \min_{\sigma' \in \mF'} \norm{\Lambda(\sigma)-\sigma'}_1$. It can be noticed that the proof in \cite{li_2018} uses only the fundamental properties of the trace norm as the base norm, and thus holds in the same way in any general probabilistic theory.


\subsection{Generalized robustness of channels}


Here, we take another approach to quantify the resourcefulness of channels, in which we aim to directly evaluate the intrinsic resourcefulness of a given channel without the aid of an underlying resource theory of states. 
Having seen the discussions for states and measurements, a natural approach to take is to
measure the resourcefulness with respect to a given set of free channels $\mO_{\mF}$.
Analogously to the cases of states and measurements, we propose the generalized robustness of channels and discuss an operational meaning of it via channel discrimination (see also Ref.\,\cite{Diaz2018usingreusing} for a discussion of this quantity in the resource theory of quantum coherence). 

The operational characterization of this measure will depend on several properties which general GPTs need not satisfy, and in particular will make heavy use of the Choi-Jamio{\l}kowski isomorphism. Although equivalent forms of this isomorphism can be obtained in GPTs beyond quantum mechanics under suitable assumptions \cite{chiribella_2010}, for the simplicity of the discussion we will limit ourselves to quantum theory.
Thus, in this subsection, operator inequalities should be understood in terms of positive semidefiniteness (we will not distinguish between the state and effect cones as they are isomorphic to each other). 

Given a convex and closed set of channels $\mO_\mF\subseteq \mT(\mV,\mV')$, we define the generalized robustness measure for channel $\Lambda\in\mT(\mV,\mV')$ with respect to $\mO_\mF$ as 
\ba
R_{\mO_{\mF}}(\Lambda) \coloneqq \min\lset r \in \mathbb{R}_+ \sbar \frac{\Lambda +  r \, \Theta}{1+r}\in \mO_{\mF}, \;\, \Theta \in \mT(\mV,\mV') \rset.
\label{eq:def robustness channels}
\ea
For further analysis, we utilize the Choi representation of channels. 
Let $J_\Lambda:= \mathbb{I}\otimes \Lambda (\dm{\tilde{\Phi}^+})$ where $\ket{\tilde{\Phi}^+}:= \sum_j \ket{jj}$ is the unnormalized maximally entangled state. 
It is well known that $\Lambda$ is completely positive iff $J_\Lambda\cgeq 0$, trace preserving iff $\Tr_{\mV'}[J_\Lambda]= \mbI_\mV$.
Let $\mO_\mF^J\subseteq \mV\otimes \mV'$ be the set of Choi matrices corresponding to channels in $\mO_\mF$.
We will assume that $\mO_\mF^J$ contains at least one interior point of $\mC \otimes \mC'$, i.e. a channel $\Gamma$ s.t. $J_\Gamma \cge 0$, so that \eqref{eq:def robustness channels} is well-defined for any channel.
Then, introducing the variable $\Xi = \Lambda + r \Theta$, \eqref{eq:def robustness channels} can be rewritten as the following optimization problem:
\ba
&{\text{\rm minimize}}& \ \ r \label{eq:opt_obj channel}\\
&{\text{\rm subject to}}&\ \ J_{\Lambda} \cleq  J_\Xi  \label{eq:opt_cond1 channel}\\
&& \ \ J_\Xi \in {\rm cone}\left(\mO_\mF^J\right)\\
&& \ \  \Tr_{\mV'}[J_\Xi]=(1+r)\,\mbI_\mV
\ea
where ${\rm cone}(\mO^J_\mF)$ is the cone generated by $\mO^J_\mF$.
The corresponding dual problem is written as (see e.g. Appendix \ref{app:duality})
\ba
&{\text{\rm maximize}}& \ \  \Tr[Y J_{\Lambda}] -1  \label{eq:dual1 obj channel}\\
&{\text{\rm subject to}}&\ \ Y=-Z+X\otimes\mbI\cgeq 0 \label{eq:dual1 cond1 channel}\\
&& \ \ X\in \mV,\,\Tr[X]=1 \label{eq:dual1 cond2 channel}\\
&& \ \ Z\in \mV\otimes \mV',\,\Tr[ZJ_\Xi]\geq 0,\ \forall J_\Xi\in \mO_\mF^J. \label{eq:dual1 cond3 channel}
\ea
It can be confirmed that the strong duality holds by taking $X=\mbI_\mV/d$ and $Z=\mbI_{\mV\mV'}/(2d)$ where $d=\dim \mV$. 
Using \eqref{eq:dual1 cond1 channel} and \eqref{eq:dual1 cond2 channel}, one can write \eqref{eq:dual1 cond3 channel} as 
\ba
 \Tr[(-Y+X\otimes\mbI)J_\Xi] &=& -\Tr[YJ_\Xi] + \Tr[\Xi(X)]\\
 &=& -\Tr[YJ_\Xi] + 1 \geq 0.
\ea

Since the objective function does not include $X$ or $Z$, we reach the following equivalent formulation;

\ba
&{\text{\rm maximize}}& \ \  \Tr[Y J_{\Lambda}] -1  \label{eq:dual2 obj channel}\\
&{\text{\rm subject to}}&\ \ Y\cgeq 0 \label{eq:dual2 cond1 channel}\\
&& \ \ \Tr[YJ_\Xi]\leq 1,\ \forall J_\Xi\in \mO_\mF^J. \label{eq:dual2 cond2 channel}
\ea

This quantity finds a connection with two different types of discrimination problems: on the one hand, when considered for a single channel, it can characterize the advantage that the channel provides over all free channels in state discrimination tasks with given input states and measurement; on the other hand, when considered for an ensemble of channels, it can be used to express the advantage provided by the ensemble in a class of channel discrimination problems over ensembles composed of free channels.

Specifically, let us first consider the problem of discriminating an ensemble of quantum states by an application of the channel $\mbI \otimes \Lambda$. The average success probability for this task can be expressed as
\begin{equation}\begin{aligned}
     p_{\rm succ}(\{p_j,\sigma_j\}, \{M_j\}, \mbI\otimes\Lambda)=\sum_j p_j\Tr\left[\mbI\otimes\Lambda(\sigma_j)M_j\right].
 \end{aligned}\end{equation}
 We then have the following result.
 \begin{thm}\label{thm:gen_rob_single_channel}
 Let $\mO_{\mF}$ be a convex and closed set of free channels. Then
 \ba
  \max_{\{p_j, \sigma_j\}, \{M_j\}}\frac{p_{\rm succ}(\{p_j,\sigma_j\},\{M_j\}, \mbI \otimes\Lambda)}{\max_{\Xi\in\mO_\mF} p_{\rm succ}(\{p_j,\sigma_j\},\{M_j\}, \mbI \otimes\Xi)} = 1 + R_{\mO_\mF}(\Lambda)
 \ea
 where the maximization is over all state ensembles $\{p_j, \sigma_j\}$ with $\sigma_j \in \Omega(\mV \otimes \mV')$ and all measurements $\{M_j\}$ in the space $\mV \otimes \mV'$.
\end{thm}
\begin{proof}
The proof again proceeds analogously to the proof of Thm.~\ref{pro:gen rob state}. Specifically, for any given measurement and state ensemble we have
\ba
 p_{\rm succ}(\{p_j,\sigma_j\},\{M_j\}, \mbI \otimes\Lambda) &=& \sum_j p_j\Tr\left[\mbI\otimes\Lambda(\sigma_j)M_j\right]\\
 &\leq& \sum_j p_j\left(1 + r\right)\Tr\left[\mbI\otimes\tilde{\Xi}(\sigma_j)M_j\right]\\
 &\leq& \left(1 + r\right) \max_{\Xi\in\mO_\mF}p_{\rm succ}(\{p_j,\sigma_j\},\{M_j\},\mbI \otimes\Xi)
 \ea
where $r = R_{\mO_\mF}(\Lambda)$ and $\tilde{\Xi}\in\mO_\mF$ is a free channel for an optimal decomposition of $\Lambda$. The converse follows by choosing the measurement $\left\{ \frac{Y}{\norm{Y}_\infty} , \mbI - \frac{Y}{\norm{Y}_\infty} \right\}$ where $Y$ is an optimal solution in \eqref{eq:dual2 obj channel} and the ensemble $\{p_i, \sigma_i\}_{i=0}^1$ defined as $p_0 = 1, \sigma_0 = \dm{\Phi^+}$ and $p_1 = 0$ where $\ket{\Phi^+}$ is the maximally entangled state and $\sigma_1$ is an arbitrary state. For this choice, it holds that
\begin{equation}\begin{aligned}
	\frac{p_{\rm succ}(\{p_j,\sigma_j\},\{M_j\}, \mbI \otimes\Lambda)}{\max_{\Xi\in\mO_\mF} p_{\rm succ}(\{p_j,\sigma_j\},\{M_j\}, \mbI \otimes\Xi)} &= \frac{\Tr \left[ \mbI \otimes \Lambda(\dm{\Phi^+}) Y\right]}{\max_{\Xi\in\mO_\mF} \Tr \left[ \mbI \otimes \Xi(\dm{\Phi^+}) Y\right]}\\
	&= \frac{\Tr \left[ J_\Lambda Y \right]}{\max_{\Xi\in\mO_\mF} \Tr \left[ J_\Xi Y \right] }\\
	&\geq 1 + R_{\mO_\mF}(\Lambda)
\end{aligned}\end{equation}
where the inequality follows from \eqref{eq:dual2 cond2 channel}.
\end{proof}

Another approach is to consider, instead of a single channel $\Lambda$, an ensemble of channels, and characterize its performance in a class of channel discrimination problems.

Consider a scenario where a channel sampled from some prior distribution is applied to one part of the input bipartite state and a collective measurement is made on the output system, where we allow for an inconclusive measurement outcome.
The success probability of this channel discrimination problem is written as
\ba
 p_{\rm succ}'(\{p_i,\mbI\otimes\Lambda_i\}_{i=0}^{N-1},\{M_i\}_{i=0}^{N},\omega)=\sum_{j=0}^{N-1} p_j\Tr\left[\left\{\mbI\otimes\Lambda_j(\omega)\right\}M_j\right]
\ea
For a given ensemble of $N-1$ channels $\{p_i,\Lambda_i\}_{i=0}^{N-1}$, where we assume that each $p_i > 0$ for simplicity, let us consider the maximum robustness for the ensemble:
 \ba
  \tilde{R}_{\mO_\mF}(\{p_i,\Lambda_i\}) \coloneqq \max_j R_{\mO_\mF}(\Lambda_j).
 \ea

Then, we obtain the following result, which connects the maximal advantage for this class of channel discrimination and the maximum robustness of channel ensembles. 

\begin{thm} \label{thm:gen_rob_channel_ensemble}
 Let $\mO_{\mF}$ be a convex and closed set of free channels. Then,
 \ba
  \max_{\substack{\omega \in \Omega(\mV\otimes\mV)\\\{M_j\}_{j=0}^N}}\frac{p_{\rm succ}'(\{p_j,\mbI\otimes\Lambda_j\}_{j=0}^{N-1},\{M_j\}_{j=0}^{N}, \omega)}{\max_{\Xi_j\in\mO_\mF}p_{\rm succ}'(\{p_j,\mbI\otimes\Xi_j\}_{j=0}^{N-1},\{M_j\}_{j=0}^{N},\omega)} = 1 + \tilde{R}_{\mO_\mF}(\{p_j,\Lambda_j\}).
  \label{eq:advantage channel average}
 \ea
\end{thm}
\begin{proof}
One direction of the inequality is shown as usual;
 \ba
 p_{\rm succ}'(\{p_j,\mbI\otimes\Lambda_j\}_{j=0}^{N-1},\{M_i\}_{j=0}^{N}, \omega) &=& \sum_{j=0}^{N-1} p_j\Tr\left[\left\{\mbI\otimes\Lambda_j(\omega)\right\}M_j\right]\\
 &\leq& \sum_{j=0}^{N-1} p_j\left(1 + r_j\right)\Tr\left[\left\{\mbI\otimes\Xi_j(\omega)\right\}M_j\right]\\
 &\leq& \left(1 + r_{j^\star}\right) \max_{\Xi_j\in\mO_\mF}p_{\rm succ}'(\{p_j,\mbI\otimes\Xi_j\}_{j=0}^{N-1},\{M_j\}_{j=0}^{N},\omega)
 \ea
 where we set $r_j\coloneqq R_{\mO_\mF}(\Lambda_j)$ and $j^\star = \operatorname{argmax}_j r_j$.
 To show the left-hand side is greater than or equal to the right-hand side, consider the state $\omega = \dm{\Phi^+}$ and the POVM defined by $M_{j^\star} = Y_{j^\star}/\|Y_{j^\star}\|_{\infty}$ where $Y_{j^\star}$ is an optimal solution for $\Lambda_{j^\star}$ in \eqref{eq:dual2 obj channel}, $M_j=0$ for $j\neq {j^\star}$ and $M_{N}=\mbI - M_{j^\star}$. 
For this choice we then have
 \ba
  \frac{p_{\rm succ}'(\{p_j,\mbI\otimes\Lambda_j\}_{j=0}^{N-1},\{M_i\}_{j=0}^{N}, \omega)}{\max_{\Xi_j\in\mO_\mF}p_{\rm succ}'(\{p_j,\mbI\otimes\Xi_j\}_{j=0}^{N-1},\{M_i\}_{j=0}^{N},\omega)} &=& \frac{\sum_{j=0}^{N-1} p_j\Tr\left[\left\{\mbI\otimes\Lambda_j(\dm{\Phi^+})\right\}M_j\right]}{\max_{\Xi_i\in\mO_\mF}\sum_{j=0}^{N-1} p_j\Tr\left[\left\{\mbI\otimes\Xi_j(\dm{\Phi^+})\right\}M_j\right]} \\
  &=& \frac{p_{j^\star}\Tr\left[J_{\Lambda_{j^\star}}Y_{j^\star}\right]}{\max_{\Xi\in\mO_\mF} p_{j^\star}\Tr\left[J_{\Xi} Y_{j^\star}\right]} \\
  &\geq& 1 + \tilde{R}_{\mO_\mF}(\{p_j,\Lambda_j\})
 \ea
 where the last inequality is due to \eqref{eq:dual2 cond2 channel}.
\end{proof}

We remark that, in the specific case of the resource theory of quantum coherence and the choice of operations $\mO_\mF$ as the so-called maximally incoherent operations (the largest set of operations preserving the set of free states)~\cite{aberg2006quantifying}, it was shown that the two approaches to the quantification of channel resources coincide and we have $P_\mF(\Lambda) = R_{\mO_\mF}(\Lambda)$ for any channel $\Lambda$ \cite{Diaz2018usingreusing}. An investigation of more general cases under which this happens is an interesting open question in the operational characterization of resources of transformations.


\section{Standard robustness of states}\label{sec:standard_rob}

We have seen that the generalized robustness is a fundamental measure capturing the resourcefulness contained in general types of resources with  operational significance.
Besides the generalized robustness, another important member of the class of robustness measures is the standard (free) robustness. 
This is also a valid resource monotone in any convex resource theory of states \cite{regula_2018} with some known operational interpretations --- the standard robustness of entanglement plays a role in activation of quantum teleportation \cite{brandao_2007} and one-shot cost for entanglement dilution under the non-entangling operations \cite{brandao_2011}, and the standard robustness of magic is related to classical simulation overhead for quantum Clifford circuits with magic-state inputs \cite{howard_2017}.  Furthermore, this quantifier in several quantum resource theories can be evaluated analytically for some states~\cite{vidal_1999,jafarizadeh_2005,howard_2017,johnston_2018} and admits computable forms as semidefinite or linear programs~\cite{howard_2017,ringbauer_2018,johnston_2018}.

However, the standard robustness has not found use in discrimination-type problems thus far, and a universal operational meaning of the standard robustness in general resource theories, whether in quantum mechanics or beyond, has not been established. 
We address this question for the standard robustness of states and indeed give such a general operational meaning in terms of the advantage for the most fundamental type of discrimination task: balanced binary channel discrimination. 

An interesting difference between the generalized robustness and standard robustness is that, unlike generalized robustness, standard robustness can diverge for any resource state in some important theories such as the resource theory of coherence and asymmetry.
It may thus appear that interpreting such a divergent quantity in an operational setting would be implausible. 
We address this issue by considering a natural figure of merit for the advantage in balanced binary channel discrimination in a way that it encompasses this seemingly-problematic circumstances as well at the limit of diverging standard robustness. 
Note that the following argument is valid in any GPT with an additional reasonable assumption introduced below.

The standard robustness of a state $\omega \in \Omega(\mV)$ with respect to a convex and closed set of free states $\mF \subseteq \Omega(\mV)$ is defined as
\ba\label{eq:std_rob_def}
  R^\mF_\mF(\omega) \coloneqq& \inf \lset r \in \mathbb{R}_+ \sbar \frac{\omega+r \pi}{1+r} \in \mF,\; \pi \in \mF \rset.
\ea
It is straightforward to verify that the standard robustness is faithful ($R_{\mF}^\mF(\omega) = 0 \iff \omega \in \mF$), convex, and monotonic under free operations in the sense that $R_{\mF}^\mF(\Lambda(\omega)) \leq R_{\mF}^\mF(\omega)$ for any linear map such that $\Lambda[\mF] \subseteq \mF$ \cite{regula_2018}. The dual optimization problem can be obtained as
\begin{equation}\begin{aligned}
{\text{\rm maximize }}& \ \  \< X, \omega \> - 1  \label{eq:dual obj standard}\\
{\text{\rm subject to }}&\ \ 0\leq \< X, \sigma \> \leq 1\ \forall \sigma\in \mF.
\end{aligned}\end{equation}
We stress that this quantity is finite for any $\omega$ only when the set $\mF$ spans the whole space $\mathcal{V}$; we will therefore make this assumption for the discussed quantities to be well-defined, although we will later see that it is not necessary for the operational interpretation of $R_\mF^\mF$. Under this assumption, the infimum in \eqref{eq:std_rob_def} is achieved, and strong duality always holds --- indeed, the cone generated by the set $\mF$ is pointed by the pointedness of $\mC$, and the dual of any pointed cone is generating, hence it contains an interior point and so we can always choose a strictly feasible solution. 

We shall show that the standard robustness characterizes the maximum advantage that a resource state provides over free states in channel discrimination. For this task to be well-defined, one needs to specify what constitutes the set of physically allowed state transformations $\mT(\mV,\mV')$ or channels, in the given GPT. As discussed before, the very basic assumption one can make about the channels is that they preserve the state cone as well as the normalization of states. Throughout this section, we will make an additional assumption:
\begin{itemize}
\item Measure-and-prepare channels of the form $\Lambda(\omega) = \sum_i \< M_i, \omega \> \omega'_i$ with $\{M_i\}_i$ being a measurement and $\{\omega'_i\}_i$ a collection of states in the output space are allowed physical transformations.
\end{itemize}
The justification for allowing measure-and-prepare channels follows from the fact that measurement and state preparation can be considered as the two building blocks of any GPT and are necessarily physically implementable~\cite{davies_1970,ludwig_1983,lami_2018-2}; one then only needs to allow for classical information about the measurement outcome to be transferred in order to implement any such measure-and-prepare channel. In particular, this assumption is clearly satisfied in any theory which extends quantum mechanics or even classical probability theory. We will denote by $\mT(\mV,\mV')$ the set of all physically allowed operations within the given GPT, and will assume throughout this section that it satisfies the assumption above. In fact, for the purpose of our proof, it will suffice to assume that only binary measure-and-prepare channels based on two-outcome measurements are allowed.

Consider now a binary channel discrimination problem where one of two channels is applied according to the ensemble $\{1/2,\Lambda_i\}_{i=0}^1$, with each $\Lambda_i \in \mT(\mV,\mV')$. The only assumption about the output space $\mV'$ that we will make is that it contains at least two distinct states, as otherwise the task becomes trivial.
Since the two channels $\Lambda_i$ are equiprobable, a random guess gives the success probability 1/2, and thus a meaningful figure of merit of this task is how much one can increase the success probability by suitably choosing a measurement and inferring the applied channel from the measurement outcome. 
Motivated by this observation, we consider the following quantity;
\ba
 p_{\rm gain}(\{1/2,\Lambda_i\}_{i=0}^1,\{M_i\}_{i=0}^1,\omega)\coloneqq p_{\rm succ}(\{1/2,\Lambda_i\}_{i=0}^1,\{M_i\}_{i=0}^1,\omega) - \frac{1}{2}
\ea
and the one being maximized with respect to measurement;
\ba
 p_{\rm gain}(\{1/2,\Lambda_i\}_{i=0}^1,\omega) \coloneqq \max_{\{M_i\}_{i=0}^1 \in \mM} p_{\rm gain}(\{1/2,\Lambda_i\}_{i=0}^1,\{M_i\}_{i=0}^1,\omega).
\ea
Then, we obtain the following result, stating that the maximum advantage in terms of success probability gain that a resource state provides over free states is characterized exactly by the standard robustness measure.  
\begin{thm}\label{thm:st_rob}
It holds that
\ba
 \max_{\Lambda_0,\Lambda_1 \in \mT(\mV,\mV')}\frac{p_{\rm gain}(\{1/2,\Lambda_i\}_{i=0}^1,\omega)}{\max_{\sigma\in\mF}p_{\rm gain}(\{1/2,\Lambda_i\}_{i=0}^1,\sigma)} = 1+2R_{\mF}^\mF(\omega).
\ea
where the maximization is over all possible ensembles $\{1/2,\Lambda_i\}_{i=0}^1$.
\end{thm}
\begin{proof}
 By the definition of the standard robustness, there exist $\tilde{\tau},\tilde{\sigma}\in\mF$ such that $\omega=\left(1+R_{\mF}^\mF(\omega)\right)\tilde{\tau}-R_{\mF}^\mF(\omega)\,\tilde{\sigma}$.
 For any ensemble of channels $\{1/2, \Lambda_i\}_{i=0}^1$, let $\{\tilde{M}_i\}_{i=0}^1$ be a measurement maximizing the quantity $p_{\rm succ}(\{1/2,\Lambda_i\}_{i=0}^1,\{M_i\}_{i=0}^1,\omega)$. Then,
 \ba
  p_{\rm gain}(\{1/2,\Lambda_i\}_{i=0}^1,\omega) &=&\frac{1}{2}\sum_{i=0}^1\< \tilde{M}_i, \Lambda_i(\omega) \> - \frac12\\
  &=& \frac{1}{2}\left(1+R_{\mF}^\mF(\omega)\right)\sum_{i=0}^1 \< \tilde{M}_i, \Lambda_i(\tilde{\tau}) \> - \frac{1}{2}R_{\mF}^\mF(\omega)\sum_{i=0}^1 \<\tilde{M}_i, \Lambda_i(\tilde{\sigma}) \> - \frac12\\
  &=& \left(1+R_{\mF}^\mF(\omega)\right)\left(\frac{1}{2}\sum_{i=0}^1 \< \tilde{M}_i, \Lambda_i(\tilde{\tau}) \> -\frac12\right) - R_{\mF}^\mF(\omega)\left(\frac{1}{2}\sum_{i=0}^1 \<\tilde{M}_i, \Lambda_i(\tilde{\sigma}) \>-\frac12\right)\\
  &\leq& \left(1+R_{\mF}^\mF(\omega)\right)\max_{\sigma\in\mF}\max_{\{M_i\}}\left(\frac{1}{2}\sum_{i=0}^1 \<M_i, \Lambda_i(\sigma) \>-\frac12\right) \\
  &&\ \  +R_{\mF}^\mF(\omega)\,\max_{\sigma\in\mF}\max_{\{M_i\}}\left(\frac{1}{2}-\frac12\sum_{i=0}^1 \<M_i, \Lambda_i(\sigma) \>\right). \label{eq:standard upper}
 \ea
Let us change the variable in the second term by $M_0=U-M_0'$ and $M_1=U-M_1'$ where $0 \cleq_{\mC\*} M_0',M_1'\cleq_{\mC\*} U$, $M_0'+M_1'= U$.
Then, the second term becomes 
\ba
 \frac{1}{2}-\frac12\sum_{i=0}^1 \< M_i, \Lambda_i(\sigma) \> = \frac{1}{2}-\frac12\sum_{i=0}^1 \< U-M_i', \Lambda_i(\sigma) \> = \frac{1}{2}\sum_{i=0}^1 \< M_i', \Lambda_i(\sigma) \>-\frac12,
\ea
and the maximum in \eqref{eq:standard upper} is taken with respect to the new variables $\{M_i'\}$.
Hence, we get for any $\Lambda_0,\Lambda_1$,
\ba 
 p_{\rm gain}(\{1/2,\Lambda_i\}_{i=0}^1,\omega)&\leq& \left(1+2R_{\mF}^\mF(\omega)\right)\max_{\sigma\in\mF}\max_{\{M_i\}}\left(\frac{1}{2}\sum_{i=0}^1 \< M_i, \Lambda_i(\sigma) \>-\frac12\right) \\
 &=&\left(1+2R_{\mF}^\mF(\omega)\right)\max_{\sigma\in\mF}p_{\rm gain}(\{1/2,\Lambda_i\}_{i=0}^1,\sigma).
\ea

To show the opposite inequality, with the change of variables $X' \coloneqq 2 X - U$ we rewrite the optimization problem \eqref{eq:dual obj standard} as
\begin{equation}\begin{aligned}\label{eq:st_robustness_alternative}
  1 + 2 R_{\mF}^\mF(\omega) = \max \lset \< X', \omega\> \sbar - 1 \leq  \< X', \sigma \> \leq 1\ \forall \sigma\in \mF \rset
\end{aligned}\end{equation}
for any state $\omega$. 
Let $X'$ be an optimal solution in the above and $\eta_0, \eta_1 \in \Omega(\mV')$ be any two distinct states, and consider the linear maps defined by
\ba
\Lambda_0 (\omega) \coloneqq& \frac12 \< U + \frac{X'}{\norm{X'}_\Omega^\circ}, \omega \> \eta_0 + \frac12 \< U - \frac{X'}{\norm{X'}_\Omega^\circ}, \omega \> \eta_1, \label{eq:channel_1}\\
\Lambda_1 (\omega) \coloneqq& \frac12 \< U - \frac{X'}{\norm{X'}_\Omega^\circ}, \omega \> \eta_0 + \frac12 \< U + \frac{X'}{\norm{X'}_\Omega^\circ}, \omega \> \eta_1. \label{eq:channel_2}
\ea
Now, since
\ba
  \< U \pm \frac{X'}{\norm{X'}^\circ_\Omega}, \rho \> = 1 \pm \frac{\< X', \rho \>}{\max_{\pi \in \Omega(\mV)} \abs{\< X', \pi \>}} \in [0,2],
\ea
for any $\rho \in \Omega(\mV)$, the set $\left\{\frac12 \left(U + \frac{X'}{\norm{X'}^\circ_\Omega}\right), \frac12 \left(U - \frac{X'}{\norm{X'}_\Omega^\circ}\right)\right\}$ constitutes a valid measurement, and so $\Lambda_0$ and $\Lambda_1$ are measure-and-prepare channels. Hence, $\Lambda_0, \Lambda_1 \in \mT(\mV,\mV')$ by our assumption about the set of allowed transformations.
We then get
\begin{equation}\begin{aligned}
  \max_{\{M_i\}}\,p_{\rm succ}(\{1/2,\Lambda_i\}_{i=0}^1,\{M_i\}_{i=0}^1,\omega) - \frac12 &= \frac14 \|\Lambda_0(\omega)-\Lambda_1(\omega) \|_\Omega \\
  &= \frac{1}{4} \norm*{\frac{\<X', \omega\>}{\norm{X'}^\circ_\Omega} (\eta_0 - \eta_1)}_\Omega\\
  &= \frac{1}{4 \norm{X'}^\circ_\Omega} \abs{\<X', \omega\>}\, \norm{\eta_0 - \eta_1}_\Omega
\end{aligned}\end{equation}
using the absolute homogeneity of the norm $\norm{\cdot}_\Omega$.
Hence, for this choice of $\Lambda_0,\Lambda_1$, we get
\ba
 \frac{p_{\rm gain}(\{1/2,\Lambda_i\}_{i=0}^1,\omega)}{\max_{\sigma\in\mF}p_{\rm gain}(\{1/2,\Lambda_i\}_{i=0}^1,\sigma)} &=& \frac{4\norm{X'}^\circ_{\Omega} \abs{\<X', \omega\>}\, \norm{\eta_0 - \eta_1}_\Omega}{4\norm{X'}^\circ_{\Omega}\, \max_{\sigma\in\mF}\abs{\<X', \sigma\>}\, \norm{\eta_0 - \eta_1}_\Omega}\\
 &=& \frac{\abs{\<X', \omega\>}}{\max_{\sigma\in\mF}\ \abs{\<X', \sigma\>}}\\
 &\geq& 1+2 R_{\mF}^\mF(\rho)
\ea
where the inequality is due to the fact that $\abs{\<X', \sigma\>} \leq 1$ by \eqref{eq:st_robustness_alternative}.
\end{proof}

Although the use of general measurements can in general provide a significant advantage over discrimination with restricted sets of measurements \cite{matthews_2009-1,lami_2017}, one can notice by following the proof of Thm. \ref{thm:st_rob} that, in the context of quantifying the advantage provided by a resource state over all free states, the restriction of measurements is inconsequential --- the standard robustness still acts as the exact quantifier of the advantage in binary channel discrimination, regardless of how the allowed measurements are restricted. Specifically, we have the following.

\begin{cor}\label{cor:st_rob any meas}
For any informationally complete closed set of measurements $\mathcal{M}_\mathcal{F}$, it holds that
\ba
 \max_{\Lambda_0,\Lambda_1 \in \mT(\mV,\mV')}\frac{\max_{\{M_i\} \in \mM_\mF} \, p_{\rm gain}(\{1/2,\Lambda_i\}_{i=0}^1, \{M_i\}_{i=0}^1, \omega)}{\max_{\{M_i\} \in \mM_\mF} \, \max_{\sigma\in\mF} \,p_{\rm gain}(\{1/2,\Lambda_i\}_{i=0}^1,\{M_i\}_{i=0}^1, \sigma)} = 1+2R_{\mF}^\mF(\omega).
\ea
where the maximization is over all possible ensembles $\{1/2,\Lambda_i\}_{i=0}^1$.
\end{cor}
\begin{proof}
This follows exactly in the same way as the proof of Thm.\,\ref{thm:st_rob} by replacing the norm $\norm{\cdot}_{\Omega}$ with the distinguishability norm $\norm{\cdot}_{\mM_\mF}$.
\end{proof}

Another fact which can be noticed from the proof of Thm.\,\ref{thm:st_rob} in Eq.\,\eqref{eq:st_robustness_alternative} is that, for any state $\omega$, the quantity $1 + 2 R_{\mF}^\mF(\omega)$ in fact corresponds to the base norm induced by the set $\mF$ \cite{reeb_2011,regula_2018}; specifically,
\begin{equation}\begin{aligned}
  1 + 2R_{\mF}^\mF(\omega) = \norm{\omega}_{\mF} = \min \lset \lambda_+ + \lambda_- \sbar \omega = \lambda_+ \sigma_+ - \lambda_- \sigma_-,\;\lambda_\pm \in \mathbb{R}_+,\;\sigma_{\pm} \in \mF \rset.
\end{aligned}\end{equation}

Although we assumed that the standard robustness takes finite value in the course of the above argument, Thm.\,\ref{thm:st_rob} and Cor.\,\ref{cor:st_rob any meas} successfully capture the cases where the standard robustness diverges for resource theories of states such as the theory of coherence and asymmetry --- the results then imply that one can always find two channels for which no free state can enable one to perform the task better than the random guess in such theories. 

We further remark an interesting connection between our results and the result in Ref.\,\cite{howard_2017} where it was found that for the case of theory of magic, the same quantity $1 + 2 R_{\mF}^\mF$ quantifies an upper bound for the overhead of classical simulation of quantum Clifford circuits with magic-state inputs. (Note that the name ``robustness of magic'' was used in Ref.\,\cite{howard_2017} to refer to the quantity $1 + 2 R_{\mF}^\mF$ instead of $R_{\mF}^\mF$.)  
It would be an interesting problem to investigate whether this is merely a coincidence, or whether the two seemingly very different tasks, channel discrimination and classical simulation of quantum circuits, are related at a deeper level through the standard robustness measure. 


\section{Complete sets of monotones}\label{sec:complete_monotones}


It is not difficult to see that resource monotones in any resource theory provide nontrivial necessary conditions on the manipulation of resources with free operations because of their monotonicity properties: 
namely, if one object contains a larger amount of resources than another with respect to any monotone, then the transformation from the less resourceful to the more resourceful one with any free operation is prohibited.   
However, a single monotone fails in most cases to completely characterize the resource transformations; it is necessary that the amount of resource measured by the monotone does not increase under free operations, but it is usually not \textit{sufficient} to ensure the existence of a free operation realizing that transformation.
Finding the necessary and sufficient conditions for the existence of a transformation by means of free operations for a given input and output object is one of the most important questions to address in resource theories, as it underlies the operational capabilities of a given resource theory.

We call a (possibly infinite) family of monotones a \textit{complete set of monotones} if it fully characterizes the necessary and sufficient conditions for the existence of a free transformation. Such sets of monotones were discussed in several specific settings~\cite{nielsen_1999,Buscemi2005clean,Matsumoto2010randomization,Buscemi2012comparison,buscemi_2012-1,horodecki_2013,Jencova2016comparison,buscemi_2016,buscemi_2017,Buscemi2017comparison,Gour2017,rosset_2018,gour_2018-2,skrzypczyk2018robustness}, but no set of general conditions with a clear operational meaning was previously known to provide a comprehensive characterization of transformations in general resource theories.

We will now show that performance of a state or measurement in a class of channel or state discrimination tasks precisely serves as a complete set of monotones for general resource theories defined in any GPT.  
This, together with the results obtained above, completes the operational characterization of quantification and exact state transformations in general resource theories in terms of discrimination tasks.


\subsection{Complete sets of monotones for states}\label{sec:complete_monotones_states}


Let $\mO \subseteq \mT(\mV,\mV)$ be a convex and closed set of transformations which contains the identity map $x \mapsto x$ and is closed under concatenation in the sense that $\Lambda_1, \Lambda_2 \in \mO$ means that $\Lambda_2 \circ \Lambda_1 \in \mO$. This set of assumptions is particularly natural if $\mO$ is identified with a set of free operations in a convex resource theory, where the action of the identity channel trivially cannot generate any resource and neither should the application of two free channels; however, our results will be completely general, and we will not need to explicitly assume any relation with resource theories.

Consider now the problem of channel discrimination with a possible inconclusive measurement outcome of channels from the set $\mO$, that is, ensembles of the form $\{p_i, \Lambda_i\}$ with each $\Lambda_i \in \mO$ and $p_i$ being the corresponding probabilities. We will first show that the performance of two given states $\omega,\omega'$ optimized over \textit{all} such discrimination problems leads to a complete set of monotones under the operations $\mO$, although as we will see later, this condition can be significantly simplified.

To begin, let us consider a fixed $(N+1)$-outcome measurement $\mbM = \{M_i\}_{i=0}^{N}$ and state $\omega$, for which the best achievable probability of success in all free channel discrimination problems is
\ba
 \tilde{p}'_{\rm succ}(\mbM,\omega) \coloneqq \max_{\substack{\{p_i, \Lambda_i\}_{i=0}^{N-1}\\\Lambda_i \in \mO}} \sum_{i=0}^{N-1} p_i \<M_i, \Lambda_i(\omega) \>.
\ea
We will additionally consider the case where the probability distribution $\{p_i\}$ is fixed a priori, and the optimization is only over the channels $\{\Lambda_i\}$ themselves; specifically, in this case we have
\ba
 \tilde{p}'_{\rm succ}\left(\mbM,\omega,\{p_i\}_{i=0}^{N-1}\right) \coloneqq \max_{\substack{\{\Lambda_i\}_{i=0}^{N-1}\\\Lambda_i \in \mO}} \sum_{i=0}^{N-1} p_i \<M_i, \Lambda_i(\omega) \>.
\ea
We now have the following result. The proof is based on an approach similar to the one taken to characterize measurement informativeness in Ref.\,\cite{skrzypczyk2018robustness} (see also Ref.\,\cite{buscemi_2016}).

\begin{thm}\label{thm:state_monotones1}
 There exists $\Lambda\in \mO$ such that $\omega'=\Lambda(\omega)$ if and only if either of the following conditions is satisfied:
 \begin{enumerate}[(i)]
 \item it holds that $\tilde{p}'_{\rm succ}(\mbM,\omega)\geq \tilde{p}'_{\rm succ}(\mbM,\omega')$ for any measurement $\mbM \in \mM$,
 \item for a fixed probability distribution $\{p_i\}_{i=0}^{N-1}$, it holds that $\tilde{p}'_{\rm succ}(\mbM,\omega,\{p_i\})\geq \tilde{p}'_{\rm succ}(\mbM,\omega',\{p_i\})$ for any measurement $\mbM = \{M_i\}_{i=0}^{N} \in \mM$. 
\end{enumerate}
\end{thm}
\begin{proof}
The proof proceeds in exactly the same way for both of the conditions. In the following, we will consider the case where the probability distribution $\{p_i\}$ can vary (case (i)), but it can equivalently be taken to be fixed (case (ii)) for the remainder of the proof.

Suppose first that $\omega'=\tilde{\Lambda}(\omega)$ with $\tilde{\Lambda} \in \mO$. Then, for any $\mbM = \{M_i\}$, we have
 \ba
  \tilde{p}'_{\rm succ}(\{M_i\},\omega')&=& \max_{\{p_i, \Lambda_i\},\,\Lambda_i \in \mO} \sum_i p_i \< M_i, \Lambda_i(\omega') \>\\
  &=& \max_{\{p_i, \Lambda_i\},\, \Lambda_i \in \mO}\sum_i p_i \<M_i, \Lambda_i \circ \tilde{\Lambda} (\omega) \>\\
  &\leq& \max_{\{p_i, \Lambda'_i\},\, \Lambda'_i \in \mO} \sum_i p_i \< M_i, \Lambda'_i (\omega) \> = \tilde{p}'_{\rm succ}(\{M_i\},\omega)
 \ea
 where the inequality is due to the closedness of $\mO$ under concatenation. 

 To show the converse, suppose that $\forall \{M_i\}$, we have $\tilde{p}'_{\rm succ}(\{M_i\},\omega)\geq \tilde{p}'_{\rm succ}(\{M_i\},\omega')$. This implies that
 \ba
  0&\leq&\inf_{\{M_i\} \in \mM} \left[\max_{\{p_i, \Lambda_i\},\, \Lambda_i \in \mO} \sum_i p_i \< M_i, \Lambda_i(\omega) \> - \max_{\{q_i, \Theta_i\},\, \Theta_i \in \mO} \sum_i q_i \<M_i, \Theta_i (\omega') \> \right]\\
  &\leq& \inf_{\{M_i\} \in \mM} \max_{\{p_i, \Lambda_i\},\, \Lambda_i \in \mO} \sum_i p_i\< M_i, \Lambda_i(\omega)-\omega' \>\\
  &\leq& \min_{\{M_i\}_{i=0}^{N} \in \mM} \max_{\substack{\{p_i, \Lambda_i\}_{i=0}^{N-1},\, \\ \Lambda_i \in \mO}} \sum_i p_i\< M_i, \Lambda_i(\omega)-\omega' \>\\
  &=&  \max_{\substack{\{p_i, \Lambda_i\}_{i=0}^{N-1},\, \\ \Lambda_i \in \mO}} \min_{\{M_i\}_{i=0}^{N} \in \mM} \sum_i p_i \< M_i, \Lambda_i(\omega)-\omega' \> \label{eq:complete monotone positive state}
 \ea
where the second inequality is obtained by setting each $q_i = p_i$ and each $\Theta_i$ to be the identity map in the second term, the third inequality is because we restricted the minimization over $N+1$-outcome measurements where $N\geq 2$ is an arbitrary integer, and the equality is due to Sion's minimax theorem~\cite{sion} by the convexity and compactness of $\mM$, $\mO$ and the set of ensembles $\{p_i, \Lambda_i\}_{i=0}^{N-1}$ (since $N$ is a finite integer), as well as the linearity of the objective function with respect to $\Lambda_i$ and $M_i$.

To show that there exists $\Lambda\in \mO$ such that $\Lambda(\omega) = \omega'$, suppose to the contrary that it does not hold true.  
Since any $\Lambda_i$ is a physical channel and thus normalization-preserving, we get that
\ba
\< U, \Lambda_i(\omega)-\omega' \> = 0 \; \forall i.
\ea
In particular, since $\Lambda_i(\omega)-\omega' \neq 0$ by assumption, we cannot have that $\Lambda_i(\omega)-\omega' \in \mC$ as this would necessarily imply that $\<U, \Lambda_i(\omega)-\omega' \> > 0$. Therefore, by the hyperplane separation theorem~\cite{rockafellar2015convex}, for every $\Lambda_i$ there exists an effect $E_i \in \mC\*$ such that $\< E_i, \Lambda_i(\omega) - \omega' \> < 0$. We now construct an incomplete measurement $\{M_i\}_{i=0}^{N-1}$ by
\begin{equation}\begin{aligned}
  M_i \coloneqq \frac{E_i}{\norm{\sum_j E_j}_{\Omega}^\circ}
\end{aligned}\end{equation}
such that $\sum_i M_i \cleq_{\mC\*} U$, and so $\{M_i\}_{i=0}^{N-1}$ can be completed to a measurement $\{M_i\}_{i=0}^{N}$ by appending another effect $M_{N} \coloneqq U - \sum_i M_i$ which will not affect the measurement outcomes corresponding to the ensemble $\{p_i, \Lambda_i\}_{i=0}^{N-1}$. We then have for this choice of $\{M_i\}$ that
\begin{equation}\begin{aligned}
  \sum_i p_i \< M_i, \Lambda_i(\omega)-\omega' \> = \sum_i \frac{p_i}{\norm{\sum_j E_j}_{\Omega}^\circ} \< E_i, \Lambda_i(\omega) - \omega' \> < 0.
\end{aligned}\end{equation}
Since such a measurement can be constructed for any ensemble $\{p_i, \Lambda_i\}$, this is in contradiction with \eqref{eq:complete monotone positive state}, which states that there exists a choice of an ensemble $\{p_i, \Lambda_i\}$ such that for every measurement $\{M_i\}$ we have $\sum_i p_i \< M_i, \Lambda_i(\omega)-\omega' \> \geq 0$. We conclude that our original assumption must have been wrong, and there exists $\Lambda\in \mO$ such that $\Lambda(\omega) = \omega'$.
\end{proof}

The result immediately establishes a general relation between channel discrimination problems and state transformations under any fixed set of operations $\mO$. Remarkably, the freedom in choosing the probability distribution $\{p_i\}$ allows us to significantly simplify this characterization, and show that much \textit{smaller} classes of channel discrimination problems already constitute complete sets of monotones.

In particular, by taking $\{p_i\}_{i=0}^1$ with $p_0 = p_1 = \frac12$, we reduce the problem to the much more straightforward task of binary channel discrimination.
\begin{cor}\label{cor:monotones_binary}
 There exists $\Lambda\in \mO$ such that $\omega'=\Lambda(\omega)$ if and only if for any three-outcome measurement $\mbM = \{M_i\}_{i=0}^2$ it holds that $\tilde{p}'_{\rm succ}\!\left(\mbM,\omega,\left\{\frac12, \frac12\right\}\right)\geq \tilde{p}'_{\rm succ}\!\left(\mbM,\omega',\left\{\frac12, \frac12\right\}\right)$.
\end{cor}
\noindent This greatly reduces the difficulty of determining whether a free transformation between two given states exists.

Another particularly interesting case is obtained by considering the binary probability distribution defined as $p_0 = 1,\, p_1 = 0$. Although going beyond what one could consider a physically-motivated ``discrimination'' task, this allows us to obtain the following characterization of a complete set of monotones under the operations $\mO$.
\begin{cor}\label{cor:monotones_unary}
 There exists $\Lambda\in \mO$ such that $\omega'=\Lambda(\omega)$ if and only if for any $E \cgeq_{\mC\*} 0$ it holds that
 \begin{equation}\begin{aligned}\label{eq:unary_cond}
     \max_{\Lambda \in \mO} \< E, \Lambda(\omega) \> \geq \max_{\Lambda \in \mO} \< E, \Lambda(\omega') \>.
 \end{aligned}\end{equation}
\end{cor}
\begin{proof}Follows from Thm.~\ref{thm:state_monotones1} by noting that for any measurement $\mbM$ we have
\begin{equation}\begin{aligned}
    \tilde{p}'_{\rm succ}(\mbM,\omega,\{1,0\}) = \max_{\{\Lambda_i\}_{i=0}^1 \subset \mO} \< M_0, \Lambda_0(\omega) \>
\end{aligned}\end{equation}
so the conditions of the Theorem reduce to verifying whether the inequality $\tilde{p}'_{\rm succ}(\mbM,\omega,\{1,0\}) \geq \tilde{p}'_{\rm succ}(\mbM,\omega',\{1,0\})$ is satisfied for any operator $M_0 \in \mC\*$ such that there exists a valid measurement $\{M_i\}_{i=0}^{2}$ --- this can be easily verified to be precisely the set $0 \cleq_{\mC\*} M_0 \cleq_{\mC\*} U$. Without loss of generality, we can then relax the constraint $M_0 \cleq_{\mC\*} U$ as any $E \in \mC\* \!\setminus\! \{0\}$ can be renormalized as $M_0 = E/\norm{E}^\circ_\Omega$.
\end{proof}
\noindent We can further make an observation that in Cor.~\ref{cor:monotones_unary} it suffices to optimize over effects $E$ which are normalized in a suitable manner --- specifically, it suffices to verify whether Eq.~\eqref{eq:unary_cond} holds for any operator $E$ in a chosen base of the cone $\mC\*$. Within quantum mechanics, or indeed in any GPT where $\mC \cong \mC\*$, this means in particular that there exists $\Lambda\in \mO$ such that $\omega'=\Lambda(\omega)$ if and only if
 \begin{equation}\begin{aligned}\label{eq:unary_monotones_states}
     \max_{\Lambda \in \mO} \< \sigma, \Lambda(\omega) \> \geq \max_{\Lambda \in \mO} \< \sigma, \Lambda(\omega') \>
 \end{aligned}\end{equation}
 holds for any \textit{state} $\sigma$. This recovers a result of Ref.\,\cite{Gour_monotones,girard_2017} obtained with different methods in the context of quantum resource theories. We note also that the class of monotones in Eq.~\eqref{eq:unary_cond} has previously been considered in the resource theory of quantum coherence \cite{tan_2018}, albeit without an operational application to state transformations under the free operations.

We further note that in Thm.~\ref{thm:state_monotones1} and Cor.~\ref{cor:monotones_binary}, one could instead consider the tasks of subchannel discrimination from a chosen set $\mO$ of normalization non-increasing maps. The proofs proceed analogously.

Alternatively, we can establish a complete set of monotones by considering a modification of the task: we will now consider channel discrimination (without inconclusive outcomes) over all valid choices of channel ensembles, but allow for the application of a single chosen prior transformation from the set $\mO$ to the ensemble before applying the channels to be discriminated. The success probability for this task for a choice of channel ensemble $\{p_i,\Lambda_i\}$ and measurement $\{M_i\}$ is given by
\ba
 \tilde{p}_{\rm succ}(\{p_i\}, \{\Lambda_i\},\{M_i\},\omega) \coloneqq \max_{\Xi \in \mO} \sum_i p_i\< M_i, \Lambda_i\circ \Xi(\omega) \>.
\ea
Note that we have separated the probability distribution $\{p_i\}$ from the corresponding channels $\{\Lambda_i\}$, for reasons which will become clear shortly. We will now show that this success probability serves as a complete set of monotones for state transformations under the operations $\mO$ in two different ways.

\begin{thm}\label{thm:state_monotones2}
 There exists $\Lambda\in \mO$ such that $\omega'=\Lambda(\omega)$ if and only if either of the following conditions is satisfied:
 \begin{enumerate}[(i)]
 \item it holds that $\tilde{p}_{\rm succ}(\{p_i\},\{\Lambda_i\},\{M_i\},\omega)\geq \tilde{p}_{\rm succ}(\{p_i\},\{\Lambda_i\},\{M_i\},\omega')$ for all channel ensembles $\{p_i,\Lambda_i\}$ and measurements $\{M_i\}$,
 \item for a fixed set of channels $\{\Lambda_i\}_{i=0}^{N-1}$ containing the identity channel $\mathrm{id}$, it holds that $\tilde{p}_{\rm succ}(\{p_i\},\{\Lambda_i\},\{M_i\},\omega)\geq \tilde{p}_{\rm succ}(\{p_i\},\{\Lambda_i\},\{M_i\},\omega')$ for all probability distributions $\{p_i\}_{i=0}^{N-1}$ and measurements $\{M_i\}_{i=0}^{N-1}$.
\end{enumerate}
\end{thm}
Note that one could also consider subchannel discrimination in (i) instead of channel discrimination. The proof proceeds in a similar way to Thm. \ref{thm:state_monotones1} and we include it in the Appendix \ref{app:complete_monotones} for completeness.

An interesting application of Thm.\,\ref{thm:state_monotones2} is obtained by choosing the set of channels $\{\Lambda_i\}_{i=0}^1$ where $\Lambda_0 = \mathrm{id}$ and $\Lambda_1 = \Theta$ is a fixed transformation. The Theorem then gives the following.
\begin{cor}\label{cor:noise_monotones}
There exists $\Lambda\in \mO$ such that $\omega'=\Lambda(\omega)$ if and only if for all two-element probability distributions $\{p_i\}$ and two-outcome measurements $\{M_i\}$ it holds that
\begin{equation}\begin{aligned}
     \tilde{p}_{\rm succ}(\{p_i\}, \{\mathrm{id},\Theta\},\{M_i\},\omega) \geq \tilde{p}_{\rm succ}(\{p_i\}, \{\mathrm{id},\Theta\},\{M_i\},\omega')
 \end{aligned}\end{equation}
 for some a priori fixed choice of $\Theta \in \mT(\mV,\mV)$.
\end{cor}

\noindent One can interpret this scenario as detecting the noise introduced by $\Theta$ with the help of prior processing with the operations $\mO$. Remarkably, Cor.~\ref{cor:noise_monotones} then shows that a single noise model completely determines the capability of the input state $\omega$ as a noise-detecting resource aided by the operations $\mO$ --- if $\omega$ is better than $\omega'$ at detecting \textit{some} type of noise $\Theta$ for any noise strength and detection strategy, $\omega$ is more capable than $\omega'$ in detecting \textit{any other} noise introduced by a physical transformation. 
This tells us for instance that, using a standard quantum mechanical example~\cite{nielsen_2011}, only considering the family of depolarizing noise models is sufficient to assess the usefulness of a given state for all possible pre-processing assisted noise detection tasks considered here. It is notable that this non-trivial fact can be shown via the seemingly unrelated problem of resource manipulation thanks to Thm.\,\ref{thm:state_monotones2}.

We will return to the problem of state transformations in the next section, where we will consider transformations of state ensembles instead of single states.


\subsection{Complete set of monotones for measurements}\label{sec:complete_monotones_meas}


A familiar picture of transforming resources involves channels applied to states, transforming one state to another. 
However, any meaningful information processing task includes a measurement at the end, so it is reasonable to consider states, channels, and measurements as parts of a single, consolidated family.
It is therefore insightful to understand channels from an alternative dual perspective: namely, not as operations transforming states, but as operations transforming measurements.
Motivated by this observation, we extend the above consideration on the transformation of states to the transformation of measurements. We show that a similar reasoning allows for an operational characterization of measurement transformations in the context of state discrimination.  

Let $\mO_{\mE}$ be a convex and closed set of effect cone-preserving unital maps $\mV\* \to \mV\*$ which furthermore includes the identity map $E \mapsto E$ and
is closed under concatenation, i.e. if $\Gamma_1,\Gamma_2 \in \mO_{\mE}$ then $\Gamma_2 \circ \Gamma_1 \in \mO_{\mE}$. These assumptions are again particularly natural in the context of a resource theory, but this is not assumed.

We will now consider a variant of state discrimination where
instead of immediately making a measurement to discriminate the state ensemble, we apply a prior transformation to the measurement effects. We shall see that by restricting a set of allowed prior operations to the chosen set $\mO_\mE$, the success probability of this task serves as a complete set of monotones for the measurements.

Recalling that any effect cone-preserving and unital operation has a corresponding dual operation which is normalization- and state cone-preserving, this can be equivalently understood as a task of state discrimination where the transformations from the set $\lset \Lambda \sbar \Lambda\* \in \mO_{\mE} \rset$ are applied to the states. The success probability of the measurement $\mbM = \{M_i\}_i$ in distinguishing the ensemble $\mA = \{p_i, \sigma_i\}$ in this setting is then
\ba
 \tilde{p}_{\rm succ}(\mA,\mbM) = \max_{\Lambda\* \in \mO_{\mE}}\sum_i p_i \<M_i, \Lambda(\sigma_i) \>.
\ea
We now show that this success probability serves as a set of complete monotones for measurements under the free operations $\mO_{\mE}$.

\begin{thm}\label{thm:monotones_meas}
 Given measurements $\mbM = \{M_i\}_{i=0}^{N}$ and $\mbM' = \{M'_i\}_{i=0}^{N}$, there exists $\Gamma\in \mO_{\mE}$ such that $\mbM'=\Gamma(\mbM)$ if and only if for all ensembles $\mA = \{p_i, \sigma_i\}_{i=0}^{N}$ it holds that $\tilde{p}_{\rm succ}(\mA,\mbM)\geq \tilde{p}_{\rm succ}(\mA,\mbM')$.
\end{thm}
\begin{proof}
 For one direction, suppose $\mbM'=\Gamma(\mbM)$. Then for any $\mA=\{p_i,\sigma_i\}$, 
 \ba
  \tilde{p}_{\rm succ}(\mA,\mbM')&=& \max_{\Lambda\* \in \mO_{\mE}}\sum_a p_a \< \Lambda\* \circ \Gamma(M_a), \sigma_a \>\\
  &\leq& \max_{\tilde{\Lambda}\* \in \mO_{\mE}}\sum_a p_a \< M_a, \tilde{\Lambda}(\sigma_a) \> = \tilde{p}_{\rm succ}(\mA,\mbM)
 \ea
 where the inequality is due to the the closedness of $\mO_{\mE}$ under concatenation.
 To show the converse, suppose $\forall \mA,\,\tilde{p}_{\rm succ}(\mA,\mbM)\geq \tilde{p}_{\rm succ}(\mA,\mbM')$.
 It implies that
 \ba
  0&\leq&\min_\mA \left[\max_{\Lambda\* \in \mO_{\mE}}\sum_a p_a \< M_a, \Lambda(\sigma_a) \> - \max_{\Delta\* \in \mO_{\mE}} \sum_a p_a \< M_a', \Delta(\sigma_a) \>\right]\\
  &\leq& \min_\mA \max_{\Lambda\* \in \mO_{\mE}}\sum_a p_a \< \Lambda\* (M_a) - M_a', \sigma_a \> \\
  &=& \max_{\Lambda\* \in \mO_{\mE}} \min_\mA \sum_a p_a \< \Lambda\* (M_a)-M_a', \sigma_a \> \label{eq:complete monotone positive}
 \ea
where the second inequality is obtained by setting $\Delta\*$ being identity for the second term, and the equality is due to Sion's minimax theorem~\cite{sion} because of the compactness and convexity of $\mA,\,\mO_{\mE}$ and the linearity of the objective function with respect to $\Lambda$ and $p_a\sigma_a$.

We will now show that there exists $\Lambda\*\in \mO_{\mE}$ such that $\Lambda\* (M_a)-M_a' = 0\ \forall a$, thus concluding the proof. To this end, suppose to the contrary that such an operation does not exist. Since $\Lambda\*$ is unital by assumption, we get for any $\Lambda$,
\ba
\sum_a \Lambda\* (M_a)-M_a' &=& \Lambda\*(U)-U = 0.
\ea
In particular, it holds that
\begin{equation}\begin{aligned}
  \< \sum_a \Lambda\* (M_a)-M_a', \omega \> = 0 \;\, \forall \omega \in \mC
\end{aligned}\end{equation}
which implies that we cannot have $\Lambda\* (M_a)-M_a' \in \mC\*$ for all $a$, as this would necessarily mean that $\Lambda\* (M_a)-M_a'$ are identically zero. Therefore, for any choice of $\Lambda\*$ there exists an index $a^\star$ and a state $\sigma \in \Omega(\mV)$ such that
\begin{equation}\begin{aligned}
  \< \Lambda\* (M_{a^\star}) - M'_{a^\star}, \sigma \> < 0.
\end{aligned}\end{equation}
Choosing the ensemble $\{p_a, \sigma_a\}$ such that $p_a = 0$ if $a \neq {a^\star}$ and $p_{{a^\star}}=1, \sigma_{{a^\star}} = \sigma$ then gives
\begin{equation}\begin{aligned}
  \sum_a p_a \< \Lambda\* (M_a)-M_a', \sigma_a \> < 0.
\end{aligned}\end{equation}
We have therefore reached a contradiction, as \eqref{eq:complete monotone positive} says that there exists a choice of $\Lambda\*$ such that any ensemble $\{p_a,\sigma_a\}$ gives $\sum_a p_a \langle \Lambda\* (M_a)-M_a', \sigma_a \rangle \geq 0$. We conclude that our original assumption must have been wrong, and therefore there exists $\Gamma = \Lambda\* \in \mO_{\mE}$ such that $M_a'=\Gamma(M_a)\  \forall a$.
\end{proof}

By using the aforementioned dual interpretation of this task as operations applied to states rather than measurements, we can additionally obtain a complete set of monotones for the transformations between state \textit{ensembles}. Such tasks have been considered in quantum information theory in different contexts~\cite{alberti_1980,chefles_2004,reeb_2011,buscemi_2012-1,heinosaari_2015,buscemi_2017}, and indeed find use in several resource theories which employ generalizations of majorization~\cite{nielsen_1999,horodecki_2013,gour_2018-2}. To this end, we will consider two different types of tasks: one, the conclusive state discrimination just as above, with probability of success given by
\ba
 \tilde{p}_{\rm succ}(\{p_a,\sigma_a\},\{M_a\}) = \max_{\Lambda \in \mO} \sum_a p_a \<M_a, \Lambda(\sigma_a) \>
\ea
with $\mO$ being a set of operations defined as in Sec.~\ref{sec:complete_monotones_states},
and two, the inconclusive state discrimination task characterized by
\ba
 \tilde{p}'_{\rm succ}(\{p_a,\sigma_a\}_{a=0}^{N-1},\{M_a\}_{a=0}^N) = \max_{\Lambda \in \mO} \sum_{a=0}^{N-1} p_a \<M_a, \Lambda(\sigma_a) \>.
\ea
The following is then a straightforward adaptation of the concepts of Thm.~\ref{thm:monotones_meas}.
\begin{cor}\label{cor:ensembles}
Let $\{\sigma_i\}_{i=0}^{N-1}$ and $\{\sigma'_i\}_{i=0}^{N-1}$ be two collections of states. Then, there exists $\Lambda \in \mO$ such that $\sigma'_i = \Lambda(\sigma_i) \; \forall i$ if and only if either of the following conditions is satisfied:
 \begin{enumerate}[(i)]
 \item $N\geq 2$ and for all $N$-outcome measurements $\mbM$ and all probability distributions $\{p_i\}_{i=0}^{N-1}$ it holds that $\tilde{p}_{\rm succ}(\{p_i, \sigma_i\},\mbM)\geq \tilde{p}_{\rm succ}(\{p_i,\sigma'_i\},\mbM)$, 
 \item $N \geq 1$ and for all $(N\!+\!1)$-outcome measurements $\mbM=\{M_i\}_{i=0}^N$ it holds that $\tilde{p}'_{\rm succ}(\{p_i, \sigma_i\},\mbM)\geq \tilde{p}'_{\rm succ}(\{p_i,\sigma'_i\},\mbM)$, where $\{p_i\}_{i=0}^{N-1}$ is any fixed probability distribution such that $p_i > 0 \; \forall i$, which in particular can be taken to be $p_i = \frac{1}{N}$.
 \end{enumerate}
\end{cor}
We include the full proof in the Appendix \ref{app:complete_monotones} for completeness. We remark that for $N=1$, we recover Cor. \ref{cor:monotones_unary}.

\section{Conclusions}

We provided a general operational characterization of quantification and manipulation of resources, the two core concepts of resource theories, in terms of the performance of state and channel discrimination tasks.
The generality of our work is three-fold; our formulations encompass general convex resource theories associated with general types of resource objects (states, measurements, and channels), and major parts of the results are valid for general probabilistic theories beyond quantum mechanics.
In particular, we found that robustness measures play central roles in bridging the quantification of resources and the success probability in discrimination tasks, specifically establishing that the maximum advantages in classes of discrimination tasks realized by resource objects are exactly quantified by the corresponding robustness measures. 
We also characterized the manipulation of resources associated with states and measurements by considering families of discrimination tasks where their success probabilities serve as complete sets of monotones that fully characterize the transformation of resources under free operations. 
In the case of quantum mechanics, we further extended the above connections between discrimination tasks and resource-theoretic concepts to single-shot information theory.

In addition to providing fundamental insights spanning a broad class of physical theories, the results are immediately applicable to a wide range of physical resources in quantum information theory. The resource theories of coherence, entanglement, magic, athermality, asymmetry, as well as their generalizations in the form of multipartite entanglement or multi-level entanglement and coherence all fit the framework introduced herein and therefore all of our results apply to them immediately. In the case of measurements, many significant insights can be gained from studying classes of measurements such as separable, PPT, incoherent, or Pauli measurements, all of which are again special cases of the resource theories considered in this work. In the characterization of channels, we obtain results applicable on the one hand to sets of free operations in the aforementioned state-based resource theories, and on the other hand obtained an operational characterization of quantum channels which can be applied in the study of arbitrary channel-based resource theories such as the resource theory of quantum memories (non-entanglement-breaking channels).

Our results furthermore reveal interesting connections between the quantification of resources and the phenomenon of data hiding. Although data hiding has been mostly discussed in the context of the theory of entanglement, where one compares the capability of LOCC measurements (or other restricted sets such as separable or PPT measurements) to that of arbitrary measurements for state discrimination, one could consider more general data hiding procedures depending on the given physical setting. For instance, in the scenario where only one party has the ability to produce magic (so-called ``magic factory''), it is sensible to encode information in a way that Pauli measurements have less capability of decoding it than arbitrary measurements, with the data-hiding ratio characterized by the difference between the capabilities of these two sets of measurements in state discrimination.
One can further think of various other scenarios such as: one party having access to a restricted but larger set of measurements than the other (not necessarily arbitrary measurements); one party having more access to resource states or channels (not measurements) than the other, in which the data could be encoded in the form of channel discrimination tasks; or the encoded data requiring discrimination tasks more intricate than the standard binary discrimination.
Our general formulations encompass such variants, allowing for considerable flexibility of the encoding strategies.

The generality of the results also provides insights into the foundation of quantum mechanics. As we showed in this work, the operational advantage realized by any type of resource object is not unique to quantum mechanics but rather a universal phenomenon stemming solely from the convexity of the underlying cones, and thus shared by general GPTs. 
Our results in particular imply that there is no separation between the generalized robustness (a priori a geometric concept) and the advantage provided in the considered classes of discrimination tasks (explicitly operational tasks) in any GPT; therefore, one cannot hope to separate a given theory from quantum mechanics by finding a gap between these quantities. 
Our results additionally provide an experimentally accessible way of bounding the geometric resource measures as well as characterizing resource transformations in any GPT by relating them with discrimination tasks. 

An interesting problem we leave for future work is to give general operational meaning to standard robustness of measurements and channels, which would solidify the fundamental operational significance of the standard robustness measure alongside that of the generalized robustness as established in this work. 
In light of the series of results obtained in this work, it can be expected that discrimination tasks are also suitable for characterizing these measures at a high level of generality.  
Additionally, it remains to understand whether general quantitative relations between the generalized and standard robustness measures can be found, and whether there exists a way to generalize our results to more members of the robustness family besides the two we considered.
We also remark that this work raises an interesting question about a unified understanding of different operational tasks via robustness measures --- besides the discrimination tasks studied in this work, robustness measures have appeared in very different operational contexts, albeit mostly in a resource-specific fashion. 
One could then speculate that these operational settings may be deeply connected by more fundamental class of tasks whose performance is somehow characterized by the robustness measures. 

Finally, in addition to the transformation of states and measurements, one could  consider the transformation of channels realized by \textit{superchannels} \cite{Chiribella2008superchannel1,zyczkowski_2008,Chiribella2008superchannel2,jencova_2014,rosset_2018,gour_2018-1,Gour_monotones}. 
The full understanding of superchannels is still on the way, and we hope that our results will help propel this journey forward from the perspective of general resource theories. 

\section*{Note added}
During the completion of this manuscript, we became aware of the independent related works by R.~Uola et al.~\cite{uola2018quantifying} as well as by M.~Oszmaniec and T.~Biswas~\cite{Oszmaniec2019}, which obtained results similar to Theorem \ref{thm:gen_rob_meas_as_advantage} on the relation between state discrimination tasks and quantification of resources associated with measurements within quantum mechanics.

\begin{acknowledgements}
We thank Paul Skrzypczyk, Marco Piani, Ludovico Lami, and Francesco Buscemi for fruitful discussions. We are especially grateful to Shunlong Luo for the suggestion of extending the results beyond quantum mechanics. We also thank Mao Miyamoto for helpful comments on Fig.\,\ref{fig:concept}, Namit Anand for comments on the manuscript, and Seth Lloyd for making the publication of this work under open access possible. R.T. acknowledges the support of NSF, ARO, IARPA, and the Takenaka Scholarship Foundation. B.R. was supported by the National Research Foundation of Singapore Fellowship No. NRF-NRFF2016-02 and the National Research Foundation and L'Agence Nationale de la Recherche joint Project No. NRF2017-NRFANR004 VanQuTe.  
\end{acknowledgements}

\appendix
\section{Proofs of results in Sec. \ref{sec:complete_monotones}} \label{app:complete_monotones}

\begingroup
\renewcommand\thethm{\ref{thm:state_monotones2}}
\begin{thm}
 There exists $\Lambda\in \mO$ such that $\omega'=\Lambda(\omega)$ if and only if either of the following conditions is satisfied:
 \begin{enumerate}[(i)]
 \item it holds that $\tilde{p}_{\rm succ}(\{p_i\},\{\Lambda_i\},\{M_i\},\omega)\geq \tilde{p}_{\rm succ}(\{p_i\},\{\Lambda_i\},\{M_i\},\omega')$ for all channel ensembles $\{p_i,\Lambda_i\}$ and measurements $\{M_i\}$,
 \item for a fixed set of channels $\{\Lambda_i\}_{i=0}^{N-1}$ containing the identity channel $\mathrm{id}$, it holds that $\tilde{p}_{\rm succ}(\{p_i\},\{\Lambda_i\},\{M_i\},\omega)\geq \tilde{p}_{\rm succ}(\{p_i\},\{\Lambda_i\},\{M_i\},\omega')$ for all probability distributions $\{p_i\}_{i=0}^{N-1}$ and measurements $\{M_i\}_{i=0}^{N-1}$.
\end{enumerate}
\end{thm}
\begin{proof}
As before, in the proof we will optimize over all sets of channels $\{\Lambda_i\}$ (case (i)), but one could equivalently consider a fixed choice of $\{\Lambda_i\}$ setting $\Lambda_0=\mathrm{id}$ (case (ii)) and the proof proceeds the same.

 For the ``only if'' direction, suppose $\omega'=\Lambda(\omega)$. Then for any $\{p_i,\Lambda_i\}$ and $\{M_i\}$, 
 \ba
  \tilde{p}_{\rm succ}(\{p_i,\Lambda_i\},\{M_i\},\omega')&=& \max_{\Xi \in \mO}\sum_i p_i\< M_i, \Lambda_i\circ\Xi(\omega') \>\\
  &=& \max_{\Xi \in \mO}\sum_i p_i\< M_i, \Lambda_i\circ\Xi\circ\Lambda(\omega) \>\\
  &\leq& \max_{\Lambda' \in \mO}\sum_i p_i\< M_i, \Lambda_i\circ\Lambda'(\omega) \> = \tilde{p}_{\rm succ}(\{p_i,\Lambda_i\},\{M_i\},\omega)
 \ea
 where the inequality is due to the closedness of $\mO$ under concatenation. 

On the other hand, assuming $\tilde{p}_{\rm succ}(\{p_i,\Lambda_i\},\{M_i\},\omega)\geq \tilde{p}_{\rm succ}(\{p_i,\Lambda_i\},\{M_i\},\omega')$ holds for all $\{p_i,\Lambda_i\}$ and $\{M_i\}$ implies
 \ba
  0&\leq&\inf_{\{p_i,\Lambda_i\},\{M_i\}} \left[\max_{\Xi \in \mO}\sum_i p_i\< M_i, \Lambda_i\circ\Xi(\omega) \> - \max_{\Xi \in \mO}\sum_i p_i\< M_i, \Lambda_i\circ\Xi(\omega') \>\right]\\
  &\leq& \inf_{\{p_i,\Lambda_i\},\{M_i\}} \max_{\Xi \in \mO}\sum_i p_i\< M_i, \Lambda_i\left(\Xi(\omega)-\omega'\right) \>\\
  &\leq& \min_{\{p_i,\Lambda_i\}_{i=0}^{N-1},\{M_i\}_{i=0}^{N-1}} \max_{\Xi \in \mO}\sum_i p_i\< M_i, \Lambda_i\left(\Xi(\omega)-\omega'\right) \> \label{eq:complete monotone finite ensemble}\\
  &=&  \max_{\Xi \in \mO}\min_{\{p_i,\Lambda_i\}_{i=0}^{N-1},\{M_i\}_{i=0}^{N-1}}\sum_i p_i\< M_i, \Lambda_i\left(\Xi(\omega)-\omega'\right) \> \label{eq:complete monotone positive state2}
 \ea
where the second inequality is obtained by setting $\Xi$ as the identity for the second term, the third inequality is because we restricted the minimization over the $N$-element sets where $N\geq2$ is an arbitrary integer, and the equality is due to the minimax theorem because of the convexity and compactness of the sets of channels, measurements, and $\mO$ and the linearity of the objective function with respect to $\Xi$, $p_i\Lambda_i$, and $M_i$.
 
Suppose now that there does not exist $\Lambda \in \mO$ such that $\Lambda(\omega)=\omega'$. Since each $\Xi \in \mO$ preserves the normalization of states, we have $\< U, \Xi(\omega)-\omega' \> = 0$
which implies that $\Xi(\omega)-\omega' \notin \mC$ and so there exists an effect $E$ such that $\< E, \Xi(\omega)-\omega' \> < 0$. Take $N=2$ in \eqref{eq:complete monotone finite ensemble} and define the measurement $\{M_0, M_1 \} = \{E, U-E\}$ and the channel ensemble defined by $p_0=1$, $\Lambda_0=\mathrm{id}$ and $p_1=0$, $\Lambda_1=\Lambda'$ where $\mathrm{id}$ denotes the identity map and $\Lambda'$ is an arbitrary channel. We then get
\ba
\sum_i p_i\< M_i, \Lambda_i\left(\Xi(\omega)-\omega'\right) \>  = \< E, \Xi(\omega)-\omega' \> < 0
\ea
for any $\Xi$, which contradicts \eqref{eq:complete monotone positive state2} and so there exists $\Lambda\in \mO$ such that $\omega'=\Lambda(\omega)$.
\end{proof}
\endgroup

\begingroup
\renewcommand\thethm{\ref{cor:ensembles}}
\begin{cor}
Let $\{\sigma_i\}_{i=0}^{N-1}$ and $\{\sigma'_i\}_{i=0}^{N-1}$ be two collections of states. Then, there exists $\Lambda \in \mO$ such that $\sigma'_i = \Lambda(\sigma_i) \; \forall i$ if and only if either of the following conditions is satisfied:
 \begin{enumerate}[(i)]
 \item $N\geq 2$ and for all $N$-outcome measurements $\mbM$ and all probability distributions $\{p_i\}_{i=0}^{N-1}$ it holds that $\tilde{p}_{\rm succ}(\{p_i, \sigma_i\},\mbM)\geq \tilde{p}_{\rm succ}(\{p_i,\sigma'_i\},\mbM)$, 
 \item $N \geq 1$ and for all $(N\!+\!1)$-outcome measurements $\mbM=\{M_i\}_{i=0}^N$ it holds that $\tilde{p}'_{\rm succ}(\{p_i, \sigma_i\},\mbM)\geq \tilde{p}'_{\rm succ}(\{p_i,\sigma'_i\},\mbM)$, where $\{p_i\}_{i=0}^{N-1}$ is any fixed probability distribution such that $p_i > 0 \; \forall i$, which in particular can be taken to be $p_i = \frac{1}{N}$.
 \end{enumerate}
\end{cor}
\begin{proof}
The ``only if'' part follows analogously to Thm.~\ref{thm:state_monotones1}. For the other implication, assume the desired operation $\Lambda \in \mO$ does not exist, and consider the case (i) first. Notice that if $\tilde{p}_{\rm succ}(\{p_i, \sigma_i\},\mbM)\geq \tilde{p}_{\rm succ}(\{p_i,\sigma'_i\},\mbM)$ for all $N$-outcome measurements and all probability distributions, then
 \ba
  0&\leq&\min_{\{M_i\},\{p_i\}} \left[ \max_{\Lambda \in \mO} \sum_i p_i\< M_i, \Lambda(\sigma_i) \> - \max_{\Theta \in \mO} \sum_i p_i\< M_i, \Theta(\sigma'_i) \> \right]\\
  &\leq& \min_{\{M_i\},\{p_i\}} \max_{\Lambda \in \mO} \sum_i p_i\< M_i, \Lambda(\sigma_i)-\sigma'_i \>\\
  &=&  \max_{\Lambda \in \mO} \min_{\{M_i\},\{p_i\}} \sum_i p_i\< M_i, \Lambda(\sigma_i)-\sigma'_i \> \label{eq:complete monotone positive ensemble}
 \ea
by Sion's minimax theorem. But since for any $\Lambda \in \mO$ we have
\ba
\< U, \Lambda(\sigma_i)-\sigma'_i \> = 0 \; \forall i
\ea
and $\Lambda(\sigma_i)-\sigma'_i$ are not uniformly 0, this means that there must exist an index $i^\star$ such that $\Lambda(\sigma_{i^\star})-\sigma'_{i^\star} \notin \mC$ and so there exists an effect $E$ such that $\< E, \Lambda(\sigma_{i^\star})-\sigma'_{i^\star} \> < 0$. Choosing an arbitrary index $i^\diamond \in \{0,\ldots,N-1\} \setminus \{i^\star\}$, where such a choice is guaranteed to exist because $N \geq 2$ by assumption, we see that the choice of measurement $\mbM = \{M_i\}_{i=0}^{N-1}$ as $M_{i^\star} = E/\norm{E}_{\Omega}^\circ$, $M_{i^\diamond} = U-E/\norm{E}_{\Omega}^\circ$, $M_{i} = 0 \; \forall i \in \{0,\ldots,N-1\} \setminus \{i^\star, i^\diamond\}$ together with the probability distribution $\{p_i\}$ defined as $p_{i^\star} = 1$, $p_i = 0 \; \forall i \neq i^\star$ contradicts Eq. \eqref{eq:complete monotone positive ensemble}.

Similarly, in case (ii) we get that
 \ba
  0&\leq& \max_{\Lambda \in \mO} \min_{\{M_i\}_{i=0}^{N}} \sum_i p_i\< M_i, \Lambda(\sigma_i)-\sigma'_i \>
 \ea
 for a fixed probability distribution, 
 and choosing the measurement $\mbM = \{M_i\}_{i=0}^N$ as $M_{i^\star} = E/\norm{E}_{\Omega}^\circ$, $M_{i} = 0 \; \forall i \in \{0,\ldots,N-1\} \setminus \{i^\star\}$, with $M_{N} = U-E/\norm{E}_{\Omega}^\circ$ completes the proof.
\end{proof}
\endgroup

\section{Duality in conic optimization}\label{app:duality}

For completeness, we include a derivation of the dual form of the optimization problems which we employ throughout the manuscript. This section is based on standard arguments found e.g. in Refs.~\cite{boyd_2004,hiriart-urruty_2001}.

Consider first some real complete normed vector spaces $\mW, \mW'$ and the optimization problem whose optimal value is given by
\begin{equation}\begin{aligned}
    \inf \lset \< A, x \> \sbar \Lambda(x) = y, \; x \in \mK \rset
\label{eq:primal general}
\end{aligned}\end{equation}
where $A \in \mW\*, y \in \mW'$ are given, $\Lambda: \mW \to \mW'$ is some linear function, and $\mK \subseteq \mW$ is a closed and convex cone. All of the optimization problems considered in this work can be expressed in this form (and indeed so can any convex optimization problem over a closed and convex set), which we will see explicitly. 

We will refer to the above as the \textit{primal} problem, and to the set $\lset x \in \mW \sbar \Lambda(x) = y, \; x \in \mK \rset$ as the feasible set. Define the Lagrangian
\begin{equation}\begin{aligned}
     L(x; Q, Z) \coloneqq \< A, x \> - \< Z, \Lambda(x) - y \> - \< Q, x\>
 \end{aligned}\end{equation}
 where $Q \in \mW\*, Z \in \mW'\*$ are the so-called Lagrange multipliers. Notice now that for every $x$ such that $\Lambda(x) - y \neq 0$, there must exist a choice of $Z \in \mW'\*$ such that $\<Z, \Lambda(x) - y \> < 0$; similarly, by the strongly separating hyperplane theorem \cite{rockafellar2015convex}, for any $x \notin \mK$ there will exist a choice of $Q \in \mK\*$ such that $\< Q, x \> < 0$, where $\mK\* = \lset Y \in \mW\* \sbar \<Y, k \> \geq 0 \; \forall k \in \mK \rset$ is the cone dual to $\mK$. This allows us to write
\begin{equation}\begin{aligned}
    \sup_{\substack{Q \in \mK\*\\Z \in \mW\*}} L(x; Q, Z) = \begin{cases} \<A, x\> & \text{ if } \Lambda(x) = y \text{ and } x \in \mK\\ \infty & \text {otherwise},\end{cases}
\end{aligned}\end{equation}
which in particular means that
\begin{equation}\begin{aligned}
    p = \inf_{x \in \mW}  \sup_{\substack{Q \in \mK\*\\Z \in \mW'\*}} L(x; Q, Z).
\end{aligned}\end{equation}
The \textit{dual} problem is then defined by interchanging the minimization and maximization in the above:
\begin{equation}\begin{aligned}
    d \coloneqq \sup_{\substack{Q \in \mK\*\\Z \in \mW'\*}} \inf_{x \in \mW} \, L(x; Q, Z).
\end{aligned}\end{equation}
Noticing that
\begin{equation}\begin{aligned}
    \inf_{x \in \mW} L(x; Q, Z) &= \inf_{x \in \mW} \< A - \Lambda\*(Z) - Q, x \> + \<Z, y \>\\
    &= \begin{cases} \< Z, y \> & \text{ if } A - \Lambda\*(Z) - Q = 0\\ -\infty & \text{ otherwise} \end{cases}
\end{aligned}\end{equation}
since $\< A - \Lambda\*(Z) - Q, x \>$ is a linear function of $x$, we can equivalently write
\begin{equation}\begin{aligned}
    d &=  \sup \lset \< Z, y \> \sbar A - \Lambda\*(Z) - Q = 0,\; Q\in \mK\* \rset \\
    &= \sup \lset \< Z, y \> \sbar A - \Lambda\*(Z) \in \mK\* \rset.
\end{aligned}\end{equation}
We often refer to this form as the dual form of the primal optimization problem $p$.

It is not difficult to see that $p \geq d$ in general, a fact often called weak Lagrange duality. Crucially, Slater's theorem (see e.g. Ref.~\cite{boyd_2004,ponstein_2004}) states that if there exists a feasible solution $x$ such that $x$ is in the (relative) interior of $\mK$ --- called a strictly feasible solution --- then $p = d$. We refer to this property as Slater's condition, and the equivalence between the primal and dual problems as strong Lagrange duality.

Consider first the generalized robustness of states, which can be expressed as
\begin{equation}\begin{aligned}
    R_\mF(\omega) + 1 = \min \lset \< U, \sigma \> \sbar \sigma - \omega \in \mC,\; \sigma \in \cone(\mF) \rset.
\label{eq:robustness state app}
\end{aligned}\end{equation}
Note that \eqref{eq:primal general} is reduced to \eqref{eq:robustness state app} by choosing $\mW=\mV\oplus\mV,\ \mW'=\mV,\ \mK={\rm cone}(\mF)\oplus \mC,\ \Lambda(x_1\oplus x_2)=x_1-x_2,\ A=U\oplus 0$ and $y=\omega$; analogous forms can be obtained for the other considered measures. 
Writing the Lagrangian as $L(\omega; X, Z) = \< U, \sigma \> - \< X, \sigma - \omega \> - \< Z, \sigma \>$, following the steps above we straightforwardly obtain the dual as
\begin{equation}\begin{aligned}
    \max \lset \< X, \omega \> \sbar X \in \mC\*, U - X \in \mF\* \rset
\end{aligned}\end{equation}
as announced in Eq.~\eqref{eq:gen rob state dual}.
Analogously, the standard robustness $R^\mF_\mF$ can be obtained by changing the constraint $\sigma - \omega \in \mC$ to $\sigma - \omega \in \cone(\mF)$.

In the case of the robustness of measurement, for any $N+1$-outcome measurement $\mbM = \{M_i\}_{i=0}^{N}$ we can write
\begin{equation}\begin{aligned}
    R_{\mE_{\mF}}(\mbM) = \min \lset \lambda \sbar M_i + N_i \in \mE_\mF \; \forall i,\; N_i \in \mC\* \; \forall i,\; \lambda U - \sum_i N_i = 0_{\mV\*} \rset
\end{aligned}\end{equation}
which gives the Lagrangian as
\begin{equation}\begin{aligned}
    L(\{N_i\}, \lambda; \{\sigma_i\}, \{\delta_i\}, \eta) &= \lambda - \sum_i \< M_i + N_i, \sigma_i \> - \sum_i \< N_i, \delta_i \> - \< \lambda U - \sum_i N_i, \eta \>
    \\ &= \lambda ( 1 - \<U, \eta \> ) + \sum_i \<N_i, \eta - \sigma_i - \delta_i \> - \sum_i \< M_i, \sigma_i \>.
\end{aligned}\end{equation}
Optimizing over the Lagrange multipliers $\sigma_i \in \mE_\mF\*$, $\delta_i \in \mC$, and $\eta \in \{0_{\mV\*}\}\* = \mV$, we get the dual as in Eq.~\eqref{eq:dual_robmes}--\eqref{eq:dual_cond4}.

In the case of the robustness of channels, we have the problem as
\begin{equation}\begin{aligned}
    R_{\mO_\mF} (\Lambda) = \min \lset \lambda \sbar J_\Xi \in \cone(\mO_\mF^J),\; J_\Xi - J_\Lambda \cgeq 0,\; \Tr_{\mV'} J_\Xi - (1+\lambda) \mbI_{\mV} = 0_{\mV} \rset.
\end{aligned}\end{equation}
Writing the Lagrangian as
\begin{equation}\begin{aligned}
    L(J_\Xi, \lambda ; X, Y, Z) &= \lambda - \< Z, J_\Xi \> - \< Y, J_\Xi - J_\Lambda \> - \< X, (1+\lambda) \mbI_\mV - \Tr_{\mV'} J_\Xi \>\\
    &= \lambda ( 1 - \< X, \mbI_{\mV} \> ) - \< Z + Y - X \otimes \mbI_{\mV'}, J_\Xi \> + \< Y, J_\Lambda \> - \< X, \mbI_{\mV} \>
\end{aligned}\end{equation}
where we used that $\Tr_{\mV'} (\cdot) \* = \cdot \otimes \mbI_{\mV'}$, an optimization over the Lagrange coefficients $X \in \mV, Y \cgeq 0$, and $Z \in \mO_\mF^J \*$ gives the desired dual problem \eqref{eq:dual1 obj channel}--\eqref{eq:dual1 cond3 channel}.


\bibliographystyle{apsrmp4-2}
\bibliography{myref}

\end{document}